\documentclass{hcig}

\usepackage{amsmath} %
\usepackage{amsthm}  
\usepackage{amssymb}
\usepackage{color}
\usepackage{mathtools} 
\usepackage{multirow}
\usepackage{graphicx}
\usepackage{caption}
\usepackage[labelsep=colon]{caption} %
\usepackage{algorithm}
\usepackage{algpseudocode}
\usepackage{booktabs}
\usepackage{adjustbox}
\usepackage{subcaption}
\usepackage{tabularx}
\usepackage{svg}
\usepackage{wrapfig}
\usepackage{fontawesome}
\usepackage{pdfpages}

\usepackage{datetime}
\usepackage{url}
\usepackage{amsfonts}
\usepackage{nicefrac}
\usepackage{microtype}
\usepackage{bm}
\usepackage{mathrsfs}
\usepackage{cite}
\usepackage{enumitem}
\usepackage{hhline}
\usepackage{arydshln}
\usepackage{accents}
\usepackage{trimclip}
\usepackage{dsfont}
\usepackage{bbding}
\usepackage{setspace}
\usepackage{array, pifont}
\usepackage{makecell, multirow}
\usepackage{xcolor,colortbl}
\usepackage{afterpage}
\usepackage{pbox}
\allowdisplaybreaks
\usepackage{cleveref}
\usepackage{hyperref}
\hypersetup{
    colorlinks,
    linkcolor=jhuaccent,
    citecolor=jhuaccent,
    urlcolor=jhu
}

\definecolor{lightgreen}{rgb}{.9,1,.9}

\newcommand{\cmark}{\ding{51}}%
\newcommand{\xmark}{\ding{55}}%
\newcommand{\halfcmark}{\cmark\textsuperscript{\kern-0.6em\raisebox{-0.48ex}{\xmark}}}

\renewcommand{\algorithmiccomment}[1]{\bgroup\hfill{$\triangleright$~#1}\egroup}

\def\NoNumber#1{{\def\alglinenumber##1{}\State #1}\addtocounter{ALG@line}{-1}}

\newcolumntype{L}[1]{>{\raggedright\arraybackslash}p{#1}}
\newcolumntype{C}[1]{>{\centering\arraybackslash}p{#1}}
\newcolumntype{R}[1]{>{\raggedleft\arraybackslash}p{#1}}

\theoremstyle{plain}
\newtheorem{proposition}{Proposition}

\newtheorem{theorem}{Theorem}

\def\defn{\,\triangleq\,}
\def\argmin{\mathop{\mathsf{arg\,min}}} 
\def\argmax{\mathop{\mathsf{arg\,max}}}
\def\lim{\mathop{\mathsf{lim}}} 

\def\prox{\mathsf{prox}}
\def\log{\mathsf{log\,}}

\def\KL{\mathsf{KL}}
\def\div{\mathsf{div}_x}

\def\ebm{{\bm{e}}}

\def\sbm{{\bm{s}}}

\def\xbm{{\bm{x}}}
\def\ybm{{\bm{y}}}
\def\zbm{{\bm{z}}}
\def\zerobm{\bm{0}}

\def\xbmhat{{\widehat{\bm{x}}}}

\def\Abm{{\bm{A}}}
\def\Bbm{{\bm{B}}}

\def\Hbm{{\bm{H}}}
\def\Ibm{{\bm{I}}}

\def\Lbm{{\bm{L}}}

\def\Wbm{{\bm{W}}}

\def\mubm{{\bm{\mu}}}
\def\Sigmabm{{\bm{\Sigma}}}

\def\prior{{\mathsf{prior}}}
\def\post{{\mathsf{post}}}

\def\R{\mathbb{R}}
\def\E{\mathbb{E}}

\def\Xbf{{\mathbf{X}}}

\def\Ncal{{\mathcal{N}}}

\def\Lcal{{\mathcal{L}}}

\def\Mcal{{\mathcal{M}}}

\DeclareMathAlphabet{\mathsfit}{T1}{\sfdefault}{\mddefault}{\sldefault}

\def\Ssfit{{\mathsfit{S}}}

\def\Tsfit{\mathsfit{T}}

\def\epsilon{\varepsilon}

\def\dXt{{d\Xt}}
\def\dBt{{d\Bbm_t}}
\def\Xt{{\xbm_t}}
\def\XT{{\xbm_T}}

\def\Bt{{\Bbm_t}}

\def\sigmat{\sigma_t}
\def\alphat{{\alpha_t}}

\def\Wk{\Wbm_k}

\def\Sk{{\sbm_{k}}}
\def\Xk{{\xbm_k}}

\def\Xkpo{{\xbm_{k+1}}}

\def\sigmak{\sigma_k}

\def\pihat{\widehat{\pi}}

\def\div{\mathsf{div}}
\def\partialt{\partial_t}

\definecolor{limegreen}{rgb}{0.2, 0.8, 0.2}

\begin{document}

\thispagestyle{firstpagestyle}

\title{Generative Translation Priors: Bayesian Imaging with Cross-Modality Image Translation}

\author{Evan Bell$^{1}$, Jiaming Liu$^{2}$, Yifan Chen$^{3}$, Yu Sun$^{1,}$\textsuperscript{\Letter}}
\address{$^1$Johns Hopkins University \quad $^2$Stanford University \quad $^3$University of California, Los Angeles\\\smallskip
{\footnotesize \textsuperscript{\Letter}Corresponding author: {\color{jhu} ysun214@jh.edu}}}

\headertitle{\small\color{jhu}Generative Translation Priors}
\headerauthors{\small\color{jhu}Bell et al.}

\maketitle

\begin{abstract}{
The ability to leverage images from co-available modalities to inform target-domain reconstruction is highly desirable in imaging algorithms. In this work, we introduce \emph{Generative Translation Priors (GTP)}—a Bayesian framework that transforms diffusion-based image-to-image translation models into cross-modality image priors for ill-posed imaging inverse problems.
GTP incorporates target-domain measurements through likelihood guidance, steering the translation process toward the desired posterior distribution. The framework is grounded in a theoretical analysis of the resulting posterior dynamics, which reveals an intrinsic bias introduced by likelihood guidance. We further characterize this bias and derive a ground-truth-free formulation for its estimation, enabling it to serve as a practical metric for assessing posterior sampling quality. 
Building on this analysis, we derive two discretized GTP algorithms based on gradient and proximal likelihood guidance, respectively.
We validate GTP on computed tomography reconstruction with magnetic resonance side information, and on positron emission tomography reconstruction with computed tomography side information. 
Experiments demonstrate that GTP effectively incorporates complementary cross-modality information and achieves high-fidelity reconstruction even under severely undersampled measurements.
\bigskip\\
{\LARGE\faGithub} \,\,\, \raisebox{.25em}{Code is available on GitHub: \href{https://github.com/Hopkins-CIG/GTP}{\textbf{\texttt{github.com/Hopkins-CIG/GTP}}}}}
\end{abstract}

\section{Introduction}

Modern imaging procedures often involve multiple modalities that capture complementary aspects of the same object. 
For example, a patient may undergo both magnetic resonance imaging (MRI) and computed tomography (CT) scans to visualize different organs~\cite{Yan.etal2016sim}, while in scientific imaging, X-ray fluorescence provides elemental composition maps complementary to X-ray ptychography~\cite{Di.etal2016optimization}. 
Exploiting such inter-modal structure for image reconstruction is therefore a desired feature in imaging algorithms. 
However, realizing its full potential demands three key ingredients: \emph{(i)} an \emph{expressive} model of inter-modal dependencies capable of capturing complex relationships; 
\emph{(ii)} sufficient \emph{generality} to apply across diverse modality pairs without customized redesign; 
and \emph{(iii)} a \emph{principled framework} that integrates such dependencies into the reconstruction process in a theoretically rigorous manner. 
Despite growing interest in this direction~\cite{Arridge.etal2021overview}, existing \emph{cross-modality image reconstruction} methods often face a trade-off in satisfying all these criteria.
The objective of this work is to address all three by employing generative image translation models as a \emph{unified interface} for inter-modal dependencies, and formalizing them as \emph{expressive} image priors within a \emph{rigorous Bayesian} framework.

The aim of image reconstruction, or imaging inverse problems, is to recover an unknown image $\xbm \in \R^n$ from incomplete and noisy measurements $\ybm \in \R^m$.
Such problems are inherently \emph{ill-posed}, meaning that the measurements cannot losslessly specify the image.
The Bayesian framework addresses this by combining the measurement \emph{likelihood} $\ell(\ybm|\xbm)$ with a \emph{prior} distribution over the image $p(\xbm)$, rendering reconstruction as inferring the posterior distribution
\begin{equation}
\label{Eq:Posterior}
\pi(\xbm|\ybm) \propto \ell(\ybm|\xbm)\, p(\xbm).
\end{equation}
As the expressivity of the image prior directly determines inference quality, reconstruction methods have evolved from hand-crafted priors, such as those based on sparsity and total variation~\cite{Rudin.etal1992, Beck.Teboulle2009a}, to learning-based priors parameterized by deep neural networks, such as image denoising and artifact-removal networks~\cite{Venkatakrishnan.etal2013, Sreehari.etal2016, Romano.etal2017, Gupta2018, Liu.etal2020, Kamilov.etal2023}.
Recent works have further extended this progression through \emph{diffusion models}~\cite{Ho.etal2020, Song.etal2021score}, enabling generative modeling of more complex image distributions and facilitating posterior sampling for reconstruction~\cite{Chung.etal2023diffusion, Bouman.etal2023generative, Sun.etal2024, Wu.etal2024principled, Xu.etal2024provably, Coeurdoux.etal2024pnpsgs}. 
Despite these advances, existing Bayesian reconstruction methods are formulated from a single-modality perspective, with cross-modality scenarios largely left out of consideration.

\begin{figure*}[!t]
    \centering
    \includegraphics[width=\linewidth]{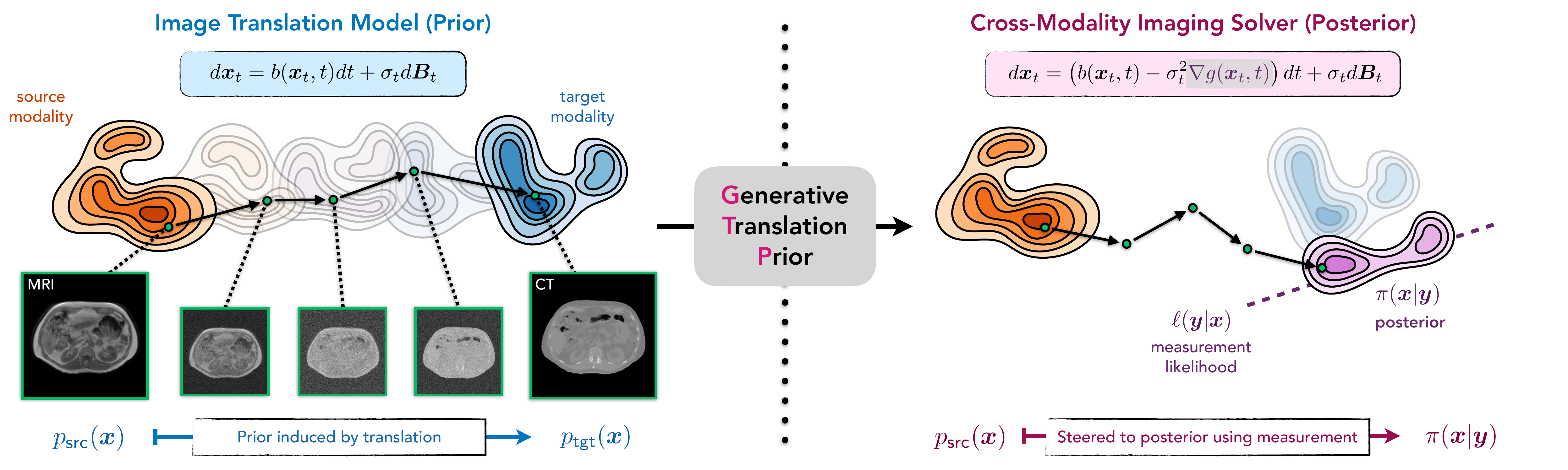}
    \caption{Schematic visualization of the proposed GTP framework. The left panel illustrates a generative image translation model that induces a target-domain prior by stochastically transporting images from a source modality to a target modality (exemplified by MRI-to-CT translation). The right panel illustrates how GTP transforms this translation prior into a posterior sampler by incorporating the measurement likelihood, steering the generative dynamics toward the target posterior distribution.}
    \label{fig:scheme}
\end{figure*}

Exploiting dependencies between co-available modality images provides an alternative paradigm for constructing image priors.
Early efforts focused on \emph{joint reconstruction} for multi-modality imaging systems~\cite{Leahy.etal1991}, where shared information across modalities is exploited through structure-promoting regularizers~\cite{Ehrhardt2023multi}. 
Despite their effectiveness, these dependency models are often hand-crafted or tailored to specific modality pairs.
Advances in \emph{image translation models} have demonstrated that deep networks can learn complex mappings between imaging modalities, providing a flexible and general mechanism for modeling inter-modal dependencies. 
Prior studies have investigated supervised end-to-end image translation~\cite{Han.etal2017mr, Oktay.etal2018attention, Dalmaz.etal2022resvit, Atalik.etal2026trust}, adversarial learning approaches for unpaired training~\cite{Isola.etal2017image, Yi.etal2017dualgan, Zhu.etal2017unpaired}, and their applications across various imaging problems~\cite{Kang.etal2019, Wang.etal2019deep, Kwon2021cycle, Xu.etal2023cross}.
Of particular interest, recent works have formulated image translation within \emph{generalized diffusion models}~\cite{Liu.etal2022, Albergo.etal2023building, Zhou.etal2024denoising, Albergo.etal2025stochastic}, interpreting translation as a probabilistic transport between image distributions and achieving state-of-the-art synthesis quality.
These advances suggest that such translation models could serve as expressive cross-modality image priors. However, a principled Bayesian formulation for integrating such models into image reconstruction remains lacking.

In this paper, we bridge this gap by proposing \emph{Generative Translation Priors (GTP)}, a Bayesian framework that integrates diffusion-based image translation with posterior sampling for cross-modality image reconstruction. 
Starting from a pre-trained translation diffusion process between two imaging modalities, GTP adapts the generative dynamics by injecting a \emph{user-defined, time-varying} likelihood term during sampling. 
This likelihood guidance steers the translation trajectory toward the desired posterior distribution while preserving the learned inter-modal dependencies (see Fig.~\ref{fig:scheme} for an illustration).
Specifically, the key contributions of this work are as follows:
\begin{itemize}
\item We characterize how the dynamics of a prior translation diffusion process, governed by a \emph{stochastic differential equation (SDE)}, are altered by the incorporation of a time-dependent likelihood term. 
In particular, by deriving the time evolution of the resulting probability density, our analysis reveals an intrinsic bias induced and controlled by the likelihood injection. 
Notably, this bias characterization requires \emph{no ground-truth information} and provides a practical criterion for evaluating and selecting among different likelihood guidance schemes to achieve the desired posterior sampling accuracy.
We further note that our analysis naturally extends to the single-modality setting, where the diffusion process evolves from a Gaussian distribution toward the image distribution.
\item We develop two GTP algorithms, termed \emph{Grad-GTP} and \emph{Prox-GTP}, which impose measurement likelihood through gradient- and proximal-based updates, respectively. 
The two algorithms also correspond to the \emph{forward} and \emph{backward Euler} discretizations of the likelihood-adapted GTP process.
We further derive the corresponding ground-truth-free bias estimators for both methods, and numerical results demonstrate a strong correlation between the estimated bias and posterior sampling accuracy.
\item We validate the proposed GTP methods on three cross-modality reconstruction tasks: sparse-view and limited-angle CT with MRI guidance, and positron emission tomography (PET) reconstruction with CT guidance. 
In the PET-CT experiment, we consider the shot noise corresponding to the Poisson likelihood.
Experimental results demonstrate that GTP algorithms achieve high-quality reconstruction across diverse inverse problems.
\end{itemize}

\section{Background}
\label{Sec:Background}

We first review fundamental concepts in Bayesian inverse problems and diffusion models within the single-modality setting. 
We then extend the discussion to image translation, cross-modality imaging, and related methods.

\subsection{Inverse Problems \& Bayesian Inference}
\label{Sec:InverseProblem}
An imaging inverse problem considers recovering an image $\xbm \in \R^n$ given the measurements $\ybm \in \R^m$ that are related by the system
\begin{equation}
\label{Eq:Inverse}
\ybm = \Abm(\xbm)+ \ebm,
\end{equation}
where the measurement operator $\Abm: \R^n \rightarrow \R^m$ models the response of the imaging system, and $\ebm \in \R^m$ represents the measurement noise.
To solve the problem, one popular Bayesian inference framework is based on the \emph{maximum a posteriori (MAP)} estimation
\begin{equation}
\label{Eq:MAP}
\xbmhat = \argmax_{\xbm\in\R^n} \ell(\ybm|\xbm) p(\xbm) 
= \argmin_{\xbm\in\R^n} \big\{ g(\xbm) + h(\xbm) \big\}.
\end{equation}
Here, $g(\xbm) = -\log \ell (\ybm|\xbm)$ is often known as the \textit{data-fidelity} term and $h(\xbm) = -\log p(\xbm)$ as the \textit{regularizer}. 
A common data-fidelity choice is the least-square loss $g(\xbm) = \frac{1}{2\beta^2}\|\ybm - \Abm(\xbm)\|_2^2$, which corresponds to the Gaussian likelihood $\ybm|\xbm \sim \Ncal(\Abm(\xbm), \beta^2I)$ with $\beta^2$ controlling the variance. 
Classic priors include the transform-domain sparsity $h(\xbm)=\tau\|\Lbm\xbm\|_1$~\cite{Candes.etal2006, Donoho2006}, which corresponds to a Laplace prior $\Lbm\xbm \sim \Lcal(0, (2/\tau^2)I)$ with $2/\tau^2$ being the variance and $\Lbm$ a transformation. 

\emph{Proximal methods}~\cite{Boyd.etal2011, Beck.Teboulle2009} provide a powerful framework for solving~\eqref{Eq:MAP} in imaging, particularly when the objective contains structured or non-differentiable terms. 
Central to these methods is the \emph{proximal operator}
\begin{equation}
\label{Eq:ProximalOperator}
\prox_{\lambda f}(\zbm) \defn \argmin_{\xbm \in \R^n} \left\{\frac{1}{2}\|\xbm-\zbm\|_2^2 + \lambda f(\xbm)\right\},
\end{equation}
where the quadratic term encourages the output $\xbm$ to remain close to the input $\zbm$, and $\lambda>0$ controls the influence of a test function $f$. 
For many linear $\Hbm$ and non-differentiable $h$, this subproblem admits a closed-form solution that can be computed efficiently~\cite{Beck.Teboulle2009, Beck.Teboulle2009a, Afonso.etal2010}.
Additionally, the proximal operator is often interpreted as a \emph{backward} gradient step, as the optimality condition of~\eqref{Eq:ProximalOperator} yields $\xbm \in \zbm - \lambda \partial f(\xbm)$, where $\partial f(\xbm)$ denotes the subdifferential at $\xbm$.

The need to characterize reconstruction uncertainty has also stimulated increasing interest in posterior sampling, \emph{i.e.} drawing $\xbm \sim \pi(\xbm|\ybm)$.
This capability is particularly useful in high-stakes applications such as clinical diagnosis~\cite{Begoli.etal2019} and scientific discovery~\cite{Kam.etal2019}. 
Early approaches primarily relied on \emph{Markov chain Monte Carlo (MCMC)} methods combined with simplified image priors, such as sparsity-promoting penalties~\cite{Bardsley2012, Repetti.etal2019}. 
However, the high dimensionality and intricate structure of natural image distributions require much more expressive priors to produce meaningful posterior samples.

\subsection{Diffusion Models as Learned Image Priors}
\label{Sec:DMPrior}
\textit{Diffusion models (DMs)}~\cite{Ho.etal2020, Song.etal2021score} have recently emerged as a transformative paradigm for modeling complex image distributions. 
Their key idea is to learn the reverse dynamics of a diffusion process that gradually maps a clean image $\xbm_0\sim p(\xbm)$ to Gaussian noise $\XT\sim\Ncal(\bm{0},\Ibm)$. 
The reverse process can be modeled by a continuous-time SDE
\begin{equation}
\label{Eq:Reverse}
\dXt = \left[ a(\Xt, t) - \sigma_t^2 \nabla \log p_t (\Xt) \right] dt + \sigmat\, \dBt,
\end{equation}
where $a:\R^n\times[0,T] \rightarrow \R^n$ denotes the \emph{drift field} inherited from the forward process, $\sigmat\in\R$ is the time-dependent diffusion coefficient, and $\Bt\in\R^n$ denotes Brownian motion.
The key quantity in~\eqref{Eq:Reverse} is the \emph{score function} $\nabla \log p_t (\Xt)$, which computes the log-density gradient of the marginal distribution associated with the intermediate state $\Xt\sim p_t(\xbm)$.
Since the true score is generally intractable, DMs learn a neural network approximation $\Ssfit_\theta(\xbm, t) \approx \nabla \log p_t(\xbm)$ using score matching techniques~\cite{Vincent.etal2011, Song.etal2019}.

The availability of a learned score approximation has motivated incorporating diffusion models as image priors in inverse problems.
One popular approach is to adapt~\eqref{Eq:Reverse} towards the posterior distribution by invoking \emph{Bayes' rule} to write the score of the time-dependent posterior $\pi_t$ as
\begin{equation}
\label{Eq:BayeScore}
\nabla \log \pi_t(\xbm_t|\ybm) = \nabla \log \ell_t(\ybm|\xbm_t) + \nabla \log p_t(\xbm_t). \nonumber
\end{equation}
This decomposition allows a pre-trained DM to provide the prior score, while the measurement model enters through the likelihood score $\nabla \log \ell_t(\ybm|\xbm_t)$. 
However, computing the time-dependent likelihood score is generally intractable~\cite{Chung.etal2023diffusion}. 
To address this challenge, existing approaches have pursued two complementary directions. 
One line of work derives approximations to $\nabla \log \ell_t(\ybm|\xbm_t)$ under simplifying assumptions~\cite{Song2022solving, Chung.etal2023diffusion, Song.etal2023pseudo, Boys.etal2024tweedie}. 
Another replaces the likelihood score with empirically designed updates that incorporate measurements directly into the diffusion process~\cite{Wang.etal2022zero, Zhu.etal2023denoising, Kawar.etal2022denoising, Song2024solving, Liu.etal2023, Choi.etal2021ilvr, Rout.etal2023solving, Chung.etal2024decomposed}. 
While these methods have demonstrated remarkable empirical success, they are typically developed from distinct algorithmic viewpoints. 
As a result, the connection between their update rules and the underlying posterior diffusion dynamics remains only partially understood, calling for a unified framework that systematically characterizes the role of likelihood injection.
For a comprehensive comparison and discussion of existing methods, we refer to~\cite{Zheng.etal2025inversebench}.

\subsection{Image Translation via Generalized Diffusion}

Image-to-image translation aims to transform an image from one domain into a corresponding image in another domain while preserving relevant underlying content.
Existing diffusion-based approaches can be broadly organized into two paradigms. 
The first adapts conventional DMs to translation tasks while retaining the standard data-to-Gaussian diffusion framework; examples include methods that learn cross-domain mappings through conditional generation or cycle-consistent training objectives~\cite{Ho.etal2021classifierfree, Ozbey.etal2023unsupervised, Xu.etal2023cyclenet, Zhang.etal2025disdiff}, and inversion-based approaches that derive a corresponding noise representation from the source image and decode it within the target domain~\cite{Meng.etal2022sdedit, Huberman.etal2024edit, Rout.etal2025semantic}.
The second paradigm is represented by \emph{generalized DMs (GDMs)}, such as \emph{diffusion bridges}~\cite{Zhou.etal2024denoising, Zheng.etal2025diffusion} and \emph{stochastic interpolants}~\cite{Albergo.etal2023building, Albergo.etal2025stochastic}, which formulate translation directly as stochastic transport between arbitrary source and target distributions. 
By removing the Gaussian endpoint, these methods provide a more principled framework for modeling cross-domain transport.

Formally, let $\xbm_0\sim p_\mathsf{src}(\xbm)$ and $\XT \sim p_\mathsf{tgt}(\xbm)$ denote the source and target distributions, respectively. 
GDMs describe the transport process through the SDE of the form
\begin{equation}
\label{Eq:GeneralSDE}
\dXt = b(\Xt, t)\,dt + \sigmat\,\dBt,
\end{equation}
where the drift field $b:\R^n\times[0,T] \rightarrow \R^n$ defines the transport dynamics between the two domains.
Different choices of the drift field recover different model instantiations within this general formulation. 
For instance, \emph{denoising diffusion bridge models (DDBM)} employ \emph{Doob's $h$-transform} to explicitly enforce both endpoint distributions~\cite{Zhou.etal2024denoising}, leading to the drift
\begin{equation}
\label{Eq:DDBMdrift}
b(\Xt, t) = a(\Xt, t) - \sigma_t^2\left(\nabla\log p(\Xt | \XT)-\nabla\log p(\XT | \Xt)\right)
\end{equation}
where $\nabla\log(\Xt | \XT)$ and $\nabla\log(\XT | \Xt)$ compute the log-density gradient with respect to $\Xt$ for the backward and forward transition probabilities, respectively.

While GDMs often achieve their best performance with paired training data, recent work has relaxed this requirement by connecting the trajectories of two independently trained GDMs~\cite{Su.etal2022dual}, exploiting the self-similarity properties of Schr\"odinger bridges~\cite{Kim.etal2024unpaired}, and learning modality-invariant latent representations that align source and target domains~\cite{Liu.etal2026unpaired}.
We note that the proposed GTP framework is fully compatible with any GDM characterized by~\eqref{Eq:GeneralSDE}, irrespective of the underlying training paradigm or level of supervision.
Additionally, as~\eqref{Eq:GeneralSDE} includes the reverse diffusion process in~\eqref{Eq:Reverse} as a special case, the GTP framework and its analysis naturally apply to standard DMs.

\subsection{Related Cross-Modality Imaging Methods}
In this section, we briefly review cross-modality imaging methods. As the field has developed from multiple perspectives, existing approaches often emphasize different aspects of the problem.

\paragraph{Joint Optimization} 
Classical cross-modality approaches often arise in multi-modality imaging systems, where images from different modalities are acquired for the same object and jointly reconstructed. 
For two modalities, a joint optimization formulation is given by
\begin{equation}
\label{Eq:JointRecon}
(\xbmhat_1,\xbmhat_2) = \argmin_{\xbm_1,\xbm_2\in\R^n} \left\{ g_1(\xbm_1) + g_2(\xbm_2) + h(\xbm_1,\xbm_2) \right\}, \nonumber
\end{equation}
where $g_1$ and $g_2$ enforce measurement fidelity for the two modalities, while $h(\xbm_1,\xbm_2)$ captures their shared structure and inter-modal dependency. 
Common choices of $h$ include patch similarity~\cite{Yan.etal2016sim}, shared edge constraints~\cite{Li.etal2016edge, Ehrhardt.etal2016, Ehrhardt.etal2016multicontrast}, and co-sparsity~\cite{Castorena.etal2016, Knoll.etal2017, Song.etal2019coupled}; see~\cite{Arridge.etal2021overview,Ehrhardt2023multi} for comprehensive reviews. 
Proximal-based algorithms have also been developed for solving such problems, in combination with alternating minimization~\cite{Yan.etal2016sim, Li.etal2016edge, Wu.etal2018spatial}, primal-dual schemes~\cite{Ehrhardt.etal2019faster, Chambolle.etal2018stochastic, Rasch.etal2018dynamic}, and Bregman distances~\cite{Rasch.etal2018joint}.
These joint optimization frameworks have been applied across a broad range of modality pairs, including MRI-CT~\cite{Yan.etal2016sim, Li.etal2016edge}, PET--CT/MRI~\cite{Ehrhardt.etal2016, Bousse.etal2016maximum, Zhang.etal2018pet}, and photoacoustic--optical coherence tomography (OCT)~\cite{Matthews.etal2017joint}.

\paragraph{Deep Learning Models}
Another major direction treats cross-modality imaging as a supervised image-to-image regression problem. 
These methods learn a direct mapping from a source modality to a target modality using paired training data, typically by minimizing pixel-wise~\cite{Gatys.etal2016} or perceptual~\cite{Johnson.etal2016perceptual} losses between the network prediction and the target image.
Early approaches commonly used convolutional encoder--decoder architectures, such as U-Net~\cite{Rxonneberger.etal2015}, to capture local spatial correspondences across modalities~\cite{Han.etal2017mr, Arabi.etal2018comparative}. 
More recent methods have incorporated attention mechanisms~\cite{Oktay.etal2018attention, Yang.etal2019show, Yao.etal2019attention}, vision transformers~\cite{Dalmaz.etal2022resvit, Deng.etal2022stytr, Peng.etal2023}, and deep unfolding architectures~\cite{Atalik.etal2026trust} to model long-range dependencies and richer cross-modality representations.
Such methods have been widely studied across different areas, including low-level computer vision~\cite{Pang.etal2022} and biomedical imaging~\cite{Dayarathna.etal2024deep}.

\paragraph{Generative Approaches}
Existing research has primarily approached cross-modality imaging as an image-to-image translation problem. 
Early methods relied on \emph{generative adversarial networks (GANs)}~\cite{Nie.etal2018medical,Yang.etal2020mri,Armanious.etal2020medgan}, while more recent studies have explored DMs for biomedical~\cite{Levac.etal2023mri,Efimov.etal2026} and scientific imaging applications~\cite{Efimov.etal2025leveraging}. 
However, these methods are designed for image synthesis and typically do not incorporate measurement information.
To the best of our knowledge, using cross-modality generative models as priors for measurement-driven reconstruction remains largely unexplored, and GTP is developed to address this gap.

\section{Generative Translation Priors}

This section introduces the GTP framework. 
We begin by constructing GTP through likelihood adaptation of a generic translation diffusion process, which reveals the presence of an intrinsic bias. 
We then establish theoretical results characterizing the impact of this bias on posterior sampling accuracy. 
Finally, we discretize the continuous-time dynamics to obtain the Grad-GTP and Prox-GTP algorithms together with their corresponding bias-estimation rules.

\subsection{Construction of GTP from Translation Diffusions}
Consider the generic translation diffusion process described in~\eqref{Eq:GeneralSDE}, where we assume the process has been calibrated such that its terminal distribution at $t = T$ coincides with the desired prior $p(\xbm)$. 
The goal of GTP is to adapt this prior translation process so that it samples from the posterior distribution $\pi(\xbm|\ybm)$. 
Since the prior translation process may not admit an explicit score representation, GTP is constructed directly at the level of marginal density evolution.

The marginal density $p_t(\xbm)$ of the prior process in~\eqref{Eq:GeneralSDE} satisfies the Fokker-Planck equation
\begin{equation}
\label{Eq:FPE}
\partialt p_t(\xbm) = -\div\Big(b(\xbm, t)\,p_t(\xbm)\Big) 
+ \frac{\sigma_t^2}{2}\Delta p_t(\xbm),
\end{equation}
where $\div = \nabla \cdot$ denotes the divergence and $\Delta$ is the Laplace operator. 
To adapt this process toward the posterior, we introduce a time-dependent likelihood potential $g(\xbm, t)$ and define the intermediate likelihood and posterior densities as
\begin{equation}
\label{Eq:Path}
\ell_t(\xbm) \propto e^{-g(\xbm, t)}, \quad 
\pi_t(\xbm) \propto e^{-g(\xbm,t)}\,p_t(\xbm),
\end{equation}
where we omit the conditioning on $\ybm$ for brevity.
To ensure $\pi_T(\xbm)=\pi(\xbm|\ybm)$, the potential must satisfy $g(\xbm,T)=-\log \ell(\ybm|\xbm)$ at the terminal time, while its evolution for $t<T$ may be chosen freely, allowing the likelihood to be introduced gradually along the diffusion path.
The following proposition characterizes how $\pi_t(\xbm)$ evolves over time.

\begin{proposition}
\label{Prop:qPDE}
Define the unnormalized posterior density (omitting the conditioning on $\ybm$ for brevity) as
\begin{equation}\nonumber
q_t(\xbm) = e^{-g(\xbm,t)}\,p_t(\xbm), \quad 
\pi_t(\xbm) = q_t(\xbm)/Z_t,
\end{equation}
where $Z_t = \int q_t(\xbm)\,d\xbm$ is the normalizing constant. 
The time evolution of $\pi_t(\xbm)$ is governed by
\begin{align}
\label{Eq:piPDE}
\partialt \pi_t(\xbm) 
&= -\div\Big(\big(b(\xbm,t) - \sigma_t^2\nabla g(\xbm,t)\big)\,\pi_t(\xbm)\Big) 
+ \frac{\sigma_t^2}{2}\Delta \pi_t(\xbm) \nonumber\\
& \quad + \pi_t(\xbm) \Big( -\partial_t g(\xbm, t) - b(\xbm, t) \cdot \nabla g(\xbm,t) \nonumber\\
& \quad - \frac{\sigma_t^2}{2}\Delta g(\xbm, t) + \frac{\sigma_t^2}{2}\|\nabla g(\xbm, t)\|_2^2 - \partial_t \log Z_t \Big). \nonumber
\end{align}
\end{proposition}
\begin{proof}
See Supplement A.1 for a detailed proof. 
\end{proof}

\noindent
Comparing Proposition~\ref{Prop:qPDE} with~\eqref{Eq:FPE} reveals that the first two terms on the right-hand side admit the Fokker-Planck structure associated with the following SDE
\begin{equation}
\label{Eq:ExAdapt}
d\xbm_t = \Big(b(\xbm_t, t) - \sigma_t^2\,\nabla g(\xbm_t, t)\Big)\,dt 
+ \sigma_t\,d\Bbm_t,
\end{equation}
where the additional drift term $-\sigmat^2 \nabla g(\xbm_t, t)$ steers the prior process toward the posterior. 
This observation naturally suggests simulating the likelihood-adapted dynamics in~\eqref{Eq:ExAdapt} for posterior sampling.
In particular, when $b(\xbm_t, t)$ is instantiated using the reverse diffusion drift, \eqref{Eq:ExAdapt} recovers the SDE underlying the DM-based methods reviewed in~\S\ref{Sec:DMPrior}. 

However, Proposition~\ref{Prop:qPDE} also contains residual terms beyond the Fokker-Planck structure.
These terms constitute the \emph{intrinsic bias} and imply that~\eqref{Eq:ExAdapt} does not, in general, realize the desired posterior density path $(\pi_t)_{t\in[0,T]}$ exactly.
Exactness is achieved only when the likelihood evolution $g(\xbm,t)$ is chosen such that the residual terms vanish; otherwise, a nonzero bias is inevitably introduced.
Our result complements recent Feynman--Kac and particle-based analyses of inference-time diffusion guidance~\cite{Skreta.etal2025feynmankac, Chen.etal2026solving, Delgadino.etal2026diffusion}, while directly characterizing the induced probability dynamics and intrinsic bias in the context of imaging inverse problems.
In the next subsection, we analyze the intrinsic bias in depth and discuss its use within the GTP framework.

\subsection{Theoretical Analysis of the Intrinsic Bias}
A careful inspection of the bias terms in Proposition~\ref{Prop:qPDE} reveals that they depend only on quantities available during sampling and therefore require no access to ground-truth posterior information. 
While it is natural to formulate an optimization problem for finding a likelihood evolution that eliminates the intrinsic bias, the resulting problem amounts to an optimal control problem that is difficult to solve.
Additionally, even evaluating the bias is computationally challenging in high dimensions due to the presence of the Laplacian term $\Delta g(\xbm,t)$. 
Rather than directly optimizing the bias, GTP takes a complementary approach: it first derives a computable formulation of the intrinsic bias and then theoretically quantifies its impact on posterior sampling accuracy. 
Together, these results establish the intrinsic bias as a \emph{viable criterion} for evaluating and selecting likelihood evolution schemes.

\begin{figure*}
\begin{minipage}[t]{.49\textwidth}
\begin{algorithm}[H]
\setstretch{1.1}
\caption{Grad-GTP}\label{Alg:GradGTP}
\begin{algorithmic}[1]
\Require $\xbm_0\in\R^n$, $\epsilon_0=0$, $\alphat=\sigmat^2/2$, $b(\xbm, t)$, $g(\xbm,t)$, and $\mathsf{bias}(\xbm, t, \frac{\sigmat^2}{2})$ in Eq.~\eqref{Eq:piBias}. 
\For{$k=0,\dots,T$}
\State $\Wk \sim \Ncal(0,I)$
\State $\Tsfit(\Xk) \leftarrow b(\Xk,k) - \frac{\sigmat^2}{2} \nabla g(\Xk, k)$
\State $\Xkpo \leftarrow \Xk + \gamma \Tsfit(\Xk) + \sigmak \sqrt{\gamma} \Wk$
\NoNumber{$\triangleright$ \textit{Bias Accumulation}}
\State $\epsilon_{k+1} \leftarrow \epsilon_k + \gamma\,\mathsf{bias}(\xbm, t, \frac{\sigmat^2}{2})$ 
\hfill \textit{Eq.~\eqref{Eq:piBias}}
\EndFor
\end{algorithmic}
\end{algorithm}%
\end{minipage}
\hspace{0.25em}
\begin{minipage}[t]{.5\textwidth}
\begin{algorithm}[H]
\setstretch{1.1}
\caption{Prox-GTP}\label{Alg:ProxGTP}
\begin{algorithmic}[1]
\Require $\xbm_0\in\R^n$, $\epsilon_0=0$, $\alphat=\sigmat^2/2$, $b(\xbm, t)$, $g(\xbm,t)$, and $\mathsf{bias}_\mathsf{prox}(\xbm, t, \frac{\sigmat^2}{2})$ in Eq.~\eqref{Eq:ProxBias}. 
\For{$k = 1, 2, \dots$}
\State $\Wk \sim \Ncal(0,I)$
\State $\Sk \leftarrow \Xk + \gamma b(\Xk,k) + \sigmak \sqrt{\gamma} \Wk$
\State $\Xkpo \leftarrow  \prox_{(\gamma\sigmat^2/2)g}(\Sk)$
\NoNumber{$\triangleright$ \textit{Bias Accumulation}}
\State $\epsilon_{k+1} \leftarrow \epsilon_k + \gamma\,\mathsf{bias}_\mathsf{prox}(\xbm, t, \frac{\sigmat^2}{2})$ \hfill \textit{Eq.~\eqref{Eq:ProxBias}}
\EndFor\label{euclidendwhile}
\end{algorithmic}
\end{algorithm}%
\end{minipage}
\end{figure*}

Proposition~\ref{Prop:qPDE} shows that the time evolution of $\pi_t(\xbm)$ can be 
decomposed into the Fokker-Planck structure of a runnable SDE and a residual intrinsic bias. 
In light of this, consider the following generalization of~\eqref{Eq:ExAdapt}
\begin{equation}
\label{Eq:brutalSDE}
d\Xt = \Big(b(\xbm_t,t) {\color{blue}-\alpha_t}\,\nabla g (\xbm_t, t)\Big)\,dt 
+ \sigma_t\,\dBt.
\end{equation}
The key difference lies in replacing the squared diffusion coefficient $\sigma_t^2$ with a free parameter $\alpha_t$, which introduces an additional degree of freedom that can be exploited to cancel the $\Delta g(\xbm,t)$ term in the bias. 
The result is summarized in the following proposition.

\begin{proposition}
\label{Prop:piPDE}
Given the reference SDE in~\eqref{Eq:brutalSDE}, the time evolution of $\pi_t(\xbm)$ can be written as
\begin{align}
\partialt \pi_t(\xbm) &= \underbrace{-\div\Big(\big(b(\xbm,t) - \alphat\nabla g(\xbm,t)\big)\,\pi_t(\xbm)\Big) + \frac{\sigma_t^2}{2}\Delta \pi_t(\xbm)}_{\text{Fokker-Planck structure of \eqref{Eq:brutalSDE}}} \nonumber\\
& \quad\, + \pi_t(\xbm)\,\Big(\mathsf{bias}(\xbm, t, \alphat) \Big), \nonumber
\end{align}
where the intrinsic bias is defined as
\begin{align}
\label{Eq:piBias}
\mathsf{bias}(\xbm, t, \alphat) = 
& - \partialt g(\xbm,t) - b(\xbm,t) \cdot \nabla g(\xbm,t) \nonumber\\
& + \left(\sigmat^2-\alphat\right)\nabla\log p_t(\xbm) \cdot \nabla g(\xbm,t)  \nonumber \\
& + \left( \frac{\sigmat^2}{2}-\alphat \right)\Delta g(\xbm,t) \nonumber\\
& + \left(\alphat-\frac{\sigmat^2}{2}\right)\|\nabla g(\xbm,t)\|_2^2 \nonumber\\
& - \partialt\log Z_t.
\end{align}
\end{proposition}
\begin{proof}
See Supplement A.2 for a detailed proof. 
\end{proof}

\noindent
Proposition~\ref{Prop:piPDE} holds for any choice of $\alpha_t$, offering flexibility in how the bias is structured. 
Setting $\alpha_t = \sigma_t^2$ recovers the bias terms in Proposition~\ref{Prop:qPDE}, where the inner product $\nabla\log p_t(\xbm) \cdot \nabla g(\xbm,t)$ vanishes. 
Setting $\alpha_t = \sigma_t^2/2$ instead retains this inner product but eliminates the Laplacian term $\Delta g(\xbm,t)$ and the gradient norm $\|\nabla g(\xbm,t)\|_2^2$, which is the computationally preferable choice. 
In either case, the inner product $\nabla\log p_t(\xbm) \cdot \nabla g(\xbm,t)$ remains computable by learning the score network $\Ssfit_\theta(\xbm, t) \approx \nabla \log p_t(\xbm)$ of the prior translation process.

We next present a theorem that quantifies the posterior sampling accuracy achieved by simulating~\eqref{Eq:brutalSDE} in terms of the intrinsic bias.
We use the \emph{Kullback-Leibler (KL)} divergence to measure the distance between two distributions, given by
\begin{equation}
\nonumber
\KL(\nu \,\|\, \pi ) = \int_{\R^n} \nu(\xbm)\,\log\frac{\nu(\xbm)}{\pi(\xbm)} \, d\xbm.
\end{equation}

\begin{theorem}
\label{Thm:BiasKL}
Let $(\pi_t)_{t\in[0,T]}$ be the posterior density path defined in~\eqref{Eq:Path}, with $\pi_T(\xbm) = \pi(\xbm|\ybm)$ as the terminal distribution.
Let $(\pihat_t)_{t\in[0,T]}$ be the marginal density path evolving according to~\eqref{Eq:brutalSDE}. 
Assuming $\pihat_0(\xbm) = \pi_0(\xbm)$, we have
\begin{align}
\KL(\pihat_T \,\|\, \pi) \leq
\left| \int_0^T \E_{\pihat_t}\big[\mathsf{bias}(\xbm,t,\alphat)\big] dt \right|. \nonumber
\end{align}
\end{theorem}
\begin{proof}
See Supplement A.3 for a detailed proof. 
\end{proof}

\noindent
As shown in Theorem~\ref{Thm:BiasKL}, the final KL error is controlled by the accumulation of the intrinsic bias along the sampling trajectory; reducing $\mathsf{bias}(\xbm,t,\alphat)$ directly yields a tighter guarantee on the posterior sampling accuracy. 
This result justifies using the intrinsic bias as a valid criterion for evaluating and selecting likelihood evolution schemes.
It is also worth noting that the right-hand side implicitly includes the contribution of the time-evolving normalizing constant through the term $\partial_t \log Z_t$. 
Since $\int_0^T \partial_t \log Z_t\,dt=\log Z_T-\log Z_0$,
this contribution is determined entirely by the endpoints and is independent of the intermediate density trajectory. 

\subsection{Discretized Algorithms: Grad-GTP \& Prox-GTP}

We now develop two practical GTP algorithms, termed Grad-GTP and Prox-GTP, by discretizing the SDE in~\eqref{Eq:brutalSDE}. 
The two algorithms differ in how the likelihood term is incorporated. 
In both cases, we fix $\alphat = \sigmat^2/2$ to enable a tractable computation of the intrinsic bias.

Grad-GTP employs a Euler–Maruyama discretization and evaluates all terms at the current iterate. 
With a stepsize $\gamma>0$, the update rule reads
\begin{align}
\label{Eq:GradGTP}
&\Xkpo = \Xk + \gamma\Tsfit(\Xk) + \sigmak\sqrt{\gamma}\,\Wk, \\
&\text{with}\quad \Tsfit(\Xk) = b(\Xk, k) - \frac{\sigmak^2}{2}\nabla g(\Xk, k), \nonumber
\end{align}
where $\Wk = \int_{k}^{k+1} \dBt$ follows the $n$-dimensional standard Gaussian distribution.
Prox-GTP, on the other hand, handles the likelihood term implicitly through a proximal update. 
The method first updates the sample using the prior drift and diffusion terms, and then applies a proximal operator associated with the likelihood
\begin{align}
\label{Eq:ProxGTP}
&\sbm_k = \xbm_k + \gamma\,b(\xbm_k, k) + \sigmak\sqrt{\gamma}\,\Wk, \\
&\Xkpo = \prox_{(\gamma\sigma_k^2/2)\,g}(\sbm_k). \nonumber
\end{align}
The proximal formulation is particularly advantageous when the forward model possesses sufficient structure to admit efficient closed-form or easily computable proximal operators. 
In the special case where the prior process is governed by the Langevin SDE, Prox-GTP coincides with the \emph{proximal Langevin algorithm}~\cite{Wibisono.etal2019proximal, Salim.etal2019stochastic, Ehrhardt.etal2024proximal}.

In both algorithms, the intrinsic bias is accumulated throughout sampling, providing a running estimate of the KL-error bound in Theorem~\ref{Thm:BiasKL}. 
While one can recompute the log-likelihood gradient in $\mathsf{bias}(\xbm,t,\alpha_t)$ for Prox-GTP, this treatment introduces repeated likelihood evaluations. 
Instead, we approximate the likelihood-gradient term by its proximal counterpart using the \emph{Moreau-envelope approximation}~\cite{Moreau1965}
\begin{equation}
\nonumber
\nabla g(\xbm,t) \approx \nabla\Mcal_{(\gamma\sigma_t^2/2)g}(\xbm) \defn \frac{2}{\gamma\sigma_t^2}\Big(\xbm - \prox_{(\gamma\sigma_t^2/2)g}(\xbm)\Big),
\end{equation}
where $\Mcal_{(\gamma\sigma_t^2/2)g}$ is the Moreau-Yosida envelope of $g$.
This substitution expresses the gradient through the proximal residual. 
Consequently, the intrinsic bias can be evaluated directly from the proximal iterates without requiring an explicit gradient of the likelihood potential. 
The corresponding proximal bias takes the form
\begin{align}
\label{Eq:ProxBias}
\mathsf{bias}_\mathsf{prox}(\xbm, t, \frac{\sigmat^2}{2}) = 
& - \partialt g(\xbm,t) - b(\xbm,t) \cdot \nabla\Mcal_{(\gamma\sigma_t^2/2)g}(\xbm) \nonumber\\
& + \frac{\sigmat^2}{2}\nabla\log p_t(\xbm) \cdot \nabla\Mcal_{(\gamma\sigma_t^2/2)g}(\xbm) \nonumber\\
& - \partialt\log Z_t.
\end{align}
Algorithms~\ref{Alg:GradGTP} and~\ref{Alg:ProxGTP} summarize the complete Grad-GTP and Prox-GTP procedures.

\section{Numerical Validations of the Bias Theory}

We first numerically validate the ability of GTP to sample from an image posterior with a known ground truth, and investigate whether the integrated intrinsic bias provides a useful reference-free criterion for assessing posterior sampling accuracy.
We focus on Prox-GTP (Algorithm~\ref{Alg:ProxGTP}) here and provide analogous results for Grad-GTP in Supplement B.

\paragraph{Experimental Setup}
We construct a Gaussian image prior $\Ncal(\mubm_\prior,\Sigmabm_\prior)$, where $\mubm_\prior$ and $\Sigmabm_\prior$ are estimated from 70,000 FFHQ images~\cite{Karras.etal2019} converted to grayscale and resized to $32\times32$.
A ground-truth image is then sampled as $\xbm^*\sim\Ncal(\mubm_\prior,\Sigmabm_\prior)$, and measurements are generated according to~\eqref{Eq:Inverse}, where $\Abm \in \R^{100\times1024}$ is a Gaussian random matrix with normalized rows.
We assume measurement noise to be Gaussian $\ebm\sim\Ncal(\zerobm,\sigma^2\Ibm)$ with $\sigma=0.01$.
With this setup, the posterior is available in closed form as $\Ncal(\mubm_\post,\Sigmabm_\post)$, where
$\mubm_\post = \Sigmabm_\post(\Sigmabm_\prior^{-1}\mubm_\prior + \frac{1}{\sigma^2}\Abm^T\ybm)$
and $\Sigmabm_\post = (\Sigmabm_\prior^{-1} + \frac{1}{\sigma^2}\Abm^T\Abm)^{-1}$.

To closely reflect practical settings, we train a DDBM as the translation prior using samples of the prior with a standard Gaussian as the source distribution.
For sampling, we use the time-varying likelihood 
\begin{equation}
\label{eq:dps_likelihood}
g(\xbm, t) = (1-t)\cdot\frac{1}{2\sigma^2}\|\Abm\hat{\xbm}_0(\xbm, t) - \ybm\|_2^2,  
\end{equation}
where $\hat{\xbm}_0(\xbm, t) \approx \E[\xbm_0|\xbm, t]$ is estimated by the trained DDBM.
This construction closely resembles the likelihood approximation used in DPS~\cite{Chung.etal2023diffusion}, and we compute its proximal by running inner gradient descent.
For DDBMs, $t=0$ corresponds to the target distribution and, by construction of the denoiser, $\hat{\xbm}_0(\xbm,0)=\xbm$.
Therefore, we have $g(\xbm,0)=\frac{1}{2\sigma^2}\|\Abm\xbm-\ybm\|_2^2$, which recovers the exact negative log-likelihood at the target endpoint as required by Theorem~\ref{Thm:BiasKL}.
The regularization strength of Prox-GTP is controlled by $\lambda$ in~\eqref{Eq:ProximalOperator}, with the proximal step size determined through $\lambda=\gamma\sigma_t^2/2$.

\begin{figure}[t!]
    \centering
    \includegraphics[width=0.65\linewidth]{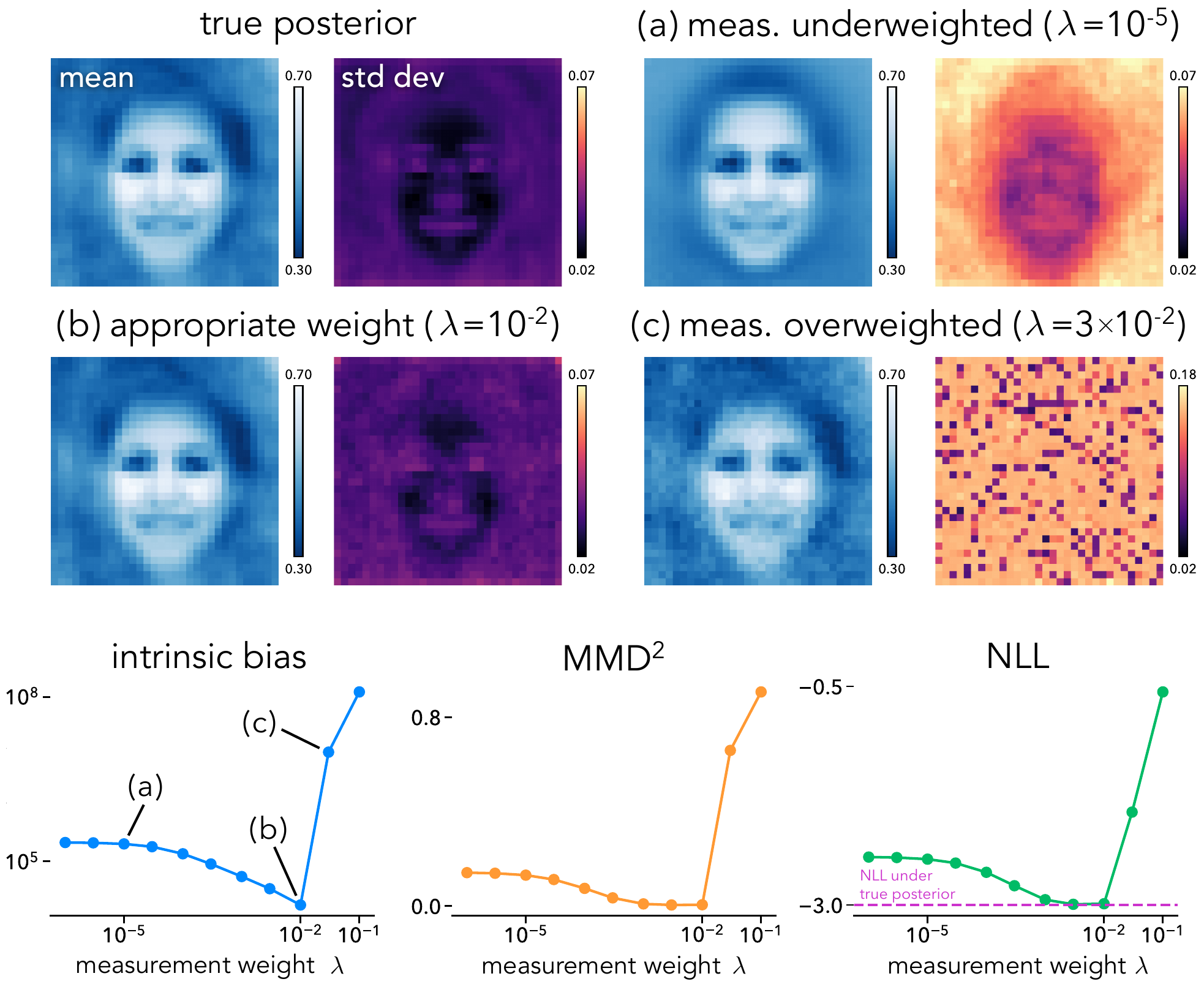}
    \caption{Numerical validation of the intrinsic bias as a reference-free criterion for assessing posterior accuracy. Prox-GTP (Algorithm~\ref{Alg:ProxGTP}) is used to recover the posterior across different likelihood strengths. 
    The recovered posterior mean and pixel-wise standard deviation are shown together with the intrinsic bias (upper bound in Theorem~\ref{Thm:BiasKL}) and two reference-based metrics: \textit{(squared) maximum mean discrepancy (MMD)} and \textit{negative log-likelihood (NLL)}. Note how the intrinsic bias closely tracks the reference-based metrics, detecting both under- and overweighting of the measurement likelihood and reaching its minimum when the recovered posterior best matches the ground truth.}
    \label{fig:gaussian_face}
\end{figure}

\paragraph{Validation Results}
\Cref{fig:gaussian_face} visualizes the mean ($\mubm$) and pixel-wise standard deviation ($\mathsf{SD}$) of the posteriors recovered by Prox-GTP across different values of $\lambda$.
As $\lambda$ increases, the figure further plots the evolution of the intrinsic bias together with two reference-based metrics for assessing posterior accuracy.
The first is the \textit{(squared) maximum mean discrepancy (MMD)}
\begin{align}
\mathsf{MMD}^2(\tilde\Xbf,\Xbf)
=&\frac{1}{m^2}\sum_{i,j=1}^m k(\tilde\xbm_i,\tilde\xbm_j)
+\frac{1}{n^2}\sum_{i,j=1}^n k(\xbm_i,\xbm_j) \nonumber\\
&-\frac{2}{mn}\sum_{i=1}^m\sum_{j=1}^n k(\tilde\xbm_i,\xbm_j), \nonumber
\end{align}
where $\tilde\Xbf=\{\tilde\xbm_i\}_{i=1}^m$ and
$\Xbf=\{\xbm_i\}_{i=1}^n$ denote samples from the recovered and true posteriors, respectively, and $k$ is a Gaussian kernel with bandwidth $\sigma_k=3.0$.
The second is the \textit{negative log-likelihood (NLL)} under an independent pixel-wise Gaussian approximation of the recovered posterior,
\begin{equation*}
\mathsf{NLL}(\xbm^*)=\frac{1}{N}\sum_{i=1}^N\left[\frac{1}{2\mathsf{SD}_{(i)}^2}(\mubm_{(i)}-\xbm_{(i)}^*)^2+\frac{1}{2}\log(2\pi\mathsf{SD}_{(i)}^2)
\right],
\end{equation*}
where $\mubm_{(i)}$ and $\mathsf{SD}_{(i)}$ are the sample mean and standard deviation of the $i$th pixel, respectively.
MMD directly measures the discrepancy between the recovered and true posterior distributions, whereas NLL evaluates how well the recovered posterior represents the ground-truth image.

As shown in~\Cref{fig:gaussian_face}, the intrinsic bias shows a trend closely consistent with both MMD and NLL as $\lambda$ varies.
For small $\lambda$, the proximal updates impose insufficient measurement information, and the recovered posterior remains close to the image prior.
This behavior is evident from both the posterior mean, which deviates from the ground-truth image, and the posterior standard deviation, which largely retains the spatial structure of the prior uncertainty.
Accordingly, MMD and NLL are high in this regime, and the intrinsic bias correctly indicates a large discrepancy.
Increasing $\lambda$ initially improves the recovered posterior by strengthening the influence of the measurement likelihood.
At an intermediate range of $\lambda$, the posterior mean better agrees with the ground truth while the posterior standard deviation retains meaningful spatial variation.
Correspondingly, all three metrics decrease and reach their lowest values in approximately the same regime.
This agreement is particularly important because computing MMD and NLL requires ground-truth information unavailable in practical inverse problems, whereas the intrinsic bias is computed without such references.
When $\lambda$ becomes too large, the proximal updates overemphasize the measurement likelihood and the recovered posterior begins to degenerate.
In particular, the sample standard deviation loses meaningful spatial structure, indicating that the recovered distribution no longer captures the uncertainty of the true posterior.
Both MMD and NLL deteriorate accordingly.
Overall, the intrinsic bias successfully detects both insufficient likelihood guidance at small $\lambda$ and excessive guidance at large $\lambda$, demonstrating the practical relevance of our bias theory.

\begin{figure*}[pt!]
    \centering
    \includegraphics[width=0.95\linewidth]{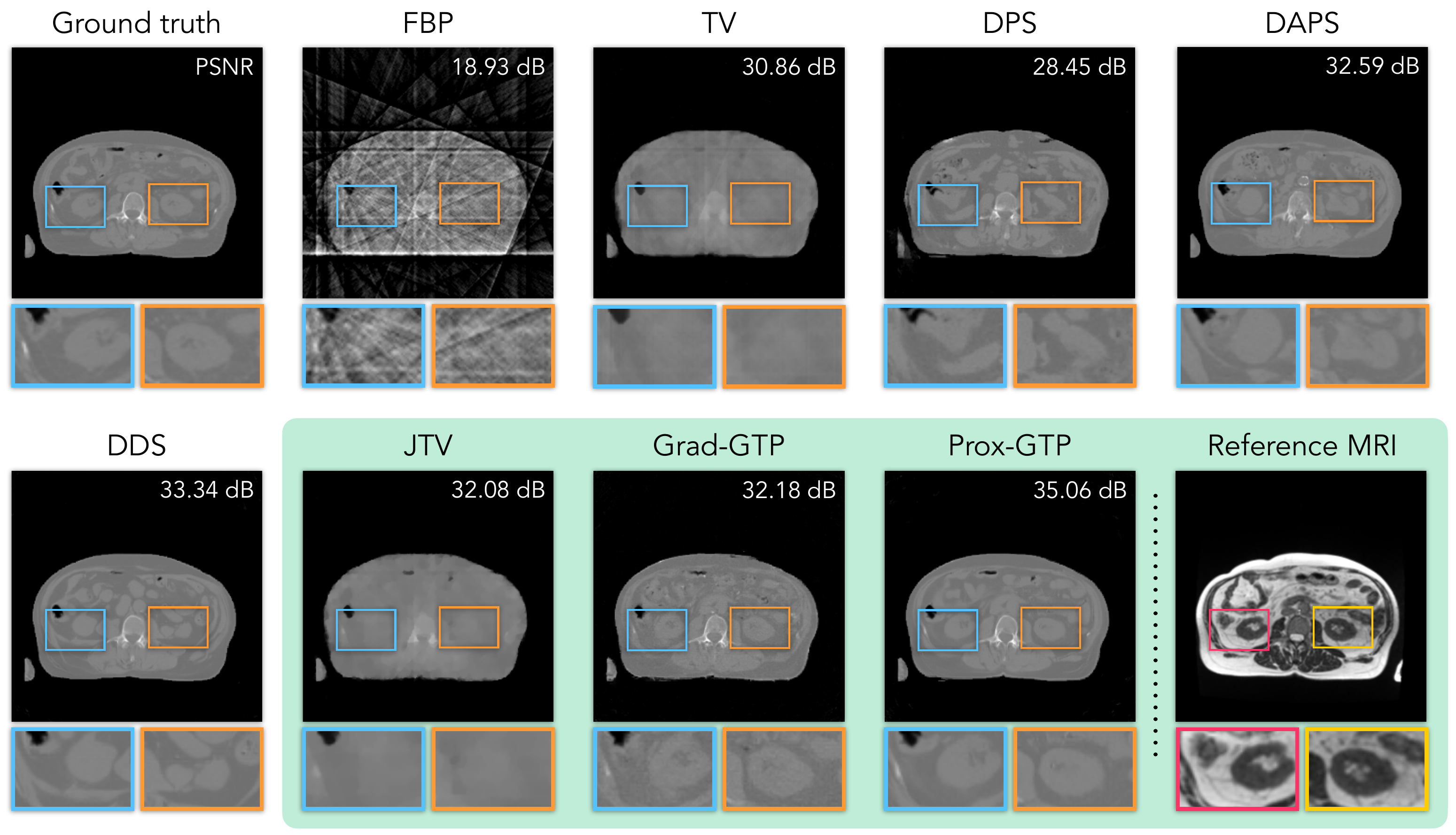}
    \caption{Visual comparison of the proposed GTP methods against baselines for 8-view SV-CT reconstruction. The shaded panels indicate methods that leverage the reference MR image during reconstruction. Zoomed-in regions highlight differences in anatomical detail. Note how the GTP methods more accurately reconstruct important anatomical structures, particularly the kidneys, compared with the baseline methods.}
    \label{fig:sv_ct}

    \vspace{2em}
    \centering
    \captionsetup{type=table}
    \small
    \begin{tabular}{lcccccccc}
    \toprule
    \multirow{2.4}{*}{Method\quad} & \multicolumn{2}{c}{4 views} & \multicolumn{2}{c}{8 views} & \multicolumn{2}{c}{16 views} & \multicolumn{2}{c}{32 views} \\ \cmidrule(lr){2-3} \cmidrule(lr){4-5} \cmidrule(lr){6-7} \cmidrule(lr){8-9}
     & PSNR & SSIM & PSNR & SSIM & PSNR & SSIM & PSNR & SSIM \\
    \midrule
    FBP & 13.81 & 0.333 & 17.70 & 0.292 & 22.28 & 0.290 & 27.41 & 0.373 \\
    TV & 25.97 & 0.858 & 30.22 & 0.901 & 34.13 & 0.932 & 37.61 & 0.953 \\
    JTV & 26.69 & \underline{0.875} & 31.26 & \textbf{0.922} & 35.91 & \textbf{0.954} & \textbf{40.03} & \textbf{0.977} \\
    DPS & 26.72 & 0.790 & 31.49 & 0.848 & 34.04 & 0.895 & 37.06 & 0.945 \\
    DAPS & 27.27 & 0.814 & 31.52 & 0.874 & 34.22 & 0.913 & 34.95 & 0.898 \\
    DDS & 28.34 & 0.835 & \underline{33.75} & 0.898 & \textbf{37.87} & 0.949 & \underline{39.91} & \underline{0.960} \\
    \cdashline{1-9}\noalign{\vskip 3pt}
    Grad-GTP & \underline{29.76} & 0.835 & 32.43 & 0.850 & 35.12 & 0.871 & 37.35 & 0.894 \\
    Prox-GTP & \textbf{30.50} & \textbf{0.885} & \textbf{34.45} & \underline{0.916} & \underline{37.81} & \underline{0.952} & 39.77 & 0.956 \\
    \bottomrule
    \end{tabular}
    \caption{Average PSNR and SSIM results attained by the proposed GTP methods and baselines for four SV-CT settings. In each column, the best result is marked in \textbf{bold}, while the second best is \underline{underlined}. For reference, the translation model used by GTP attains a PSNR of 24.94 dB and an SSIM of 0.839 without measurement guidance.}
    \label{tab:sv_ct_main}
\end{figure*}

\section{Experiments on Cross-Modality Imaging}

We now demonstrate the effectiveness of GTP for cross-modality image reconstruction across different imaging modalities and measurement models.
We first consider sparse-view and limited-angle CT reconstruction using MRI as side-modality guidance.
We then study low-dose PET reconstruction with CT guidance, which is a clinically realistic cross-modality setting since PET and CT are routinely acquired together.
This experiment further extends our evaluation beyond Gaussian measurement models by considering the Poisson noise in PET imaging.

Throughout our experiments, we instantiate the translation priors using DDBMs~\cite{Zhou.etal2024denoising}, whose dynamics naturally fit the SDE formulation in~\eqref{Eq:GeneralSDE} and~\eqref{Eq:DDBMdrift}. 
Each DDBM is trained on paired cross-modal images to model the inter-modal distribution, independently of the downstream inverse problem. 
Further details on the model architectures and training procedures, together with additional results using translation priors trained on unpaired data, are provided in Supplement C.1 and D.1.

\begin{figure*}[pt!]

    \centering
    \centering
    \includegraphics[width=0.95\linewidth]{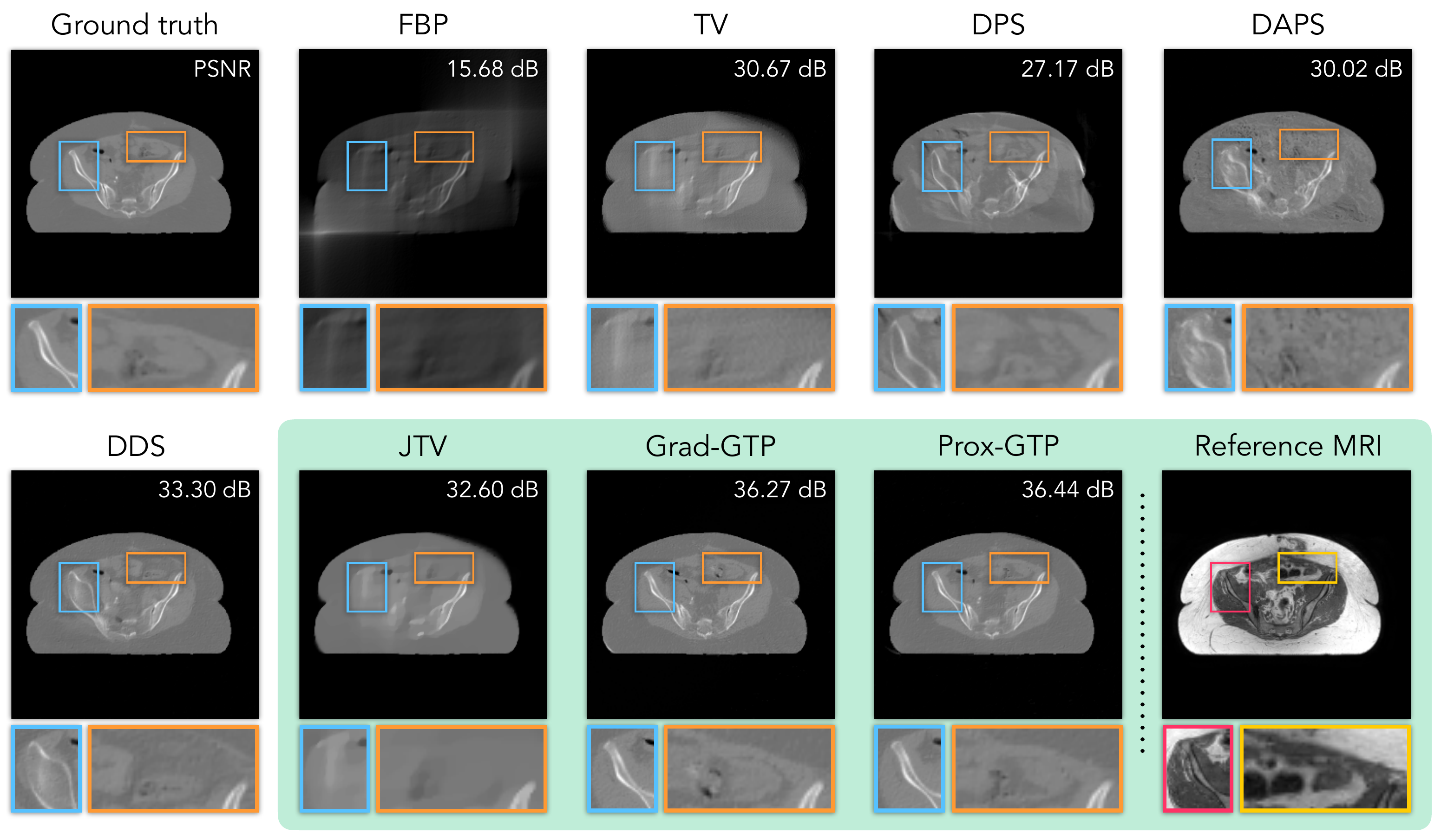}
    \caption{Visual comparison of $90^\circ$ LA-CT reconstructions obtained using the proposed GTP methods and baselines. The shaded panels indicate methods that leverage the reference MR image during reconstruction. Zoomed-in regions highlight differences in anatomical detail. Note how the GTP methods more accurately reconstruct major anatomical structures, particularly bones, that are poorly observed due to the missing angular measurements.}
    \label{fig:la_ct}

    \vspace{2em}
    \captionsetup{type=table}
    \small
    \begin{tabular}{lcccccccc}
    \toprule
    \multirow{2.4}{*}{Method\quad} & \multicolumn{2}{c}{30 degrees} & \multicolumn{2}{c}{60 degrees} & \multicolumn{2}{c}{90 degrees} & \multicolumn{2}{c}{120 degrees} \\ \cmidrule(lr){2-3} \cmidrule(lr){4-5} \cmidrule(lr){6-7} \cmidrule(lr){8-9}
     & PSNR & SSIM & PSNR & SSIM & PSNR & SSIM & PSNR & SSIM \\
    \midrule
    FBP & 12.41 & 0.118 & 13.60 & 0.158 & 16.05 & 0.267 & 20.32 & 0.314 \\
    TV & 19.88 & 0.675 & 23.60 & 0.815 & 29.49 & 0.910 & 33.58 & 0.937 \\
    JTV & 24.00 & 0.815 & 26.79 & \underline{0.867} & 30.81 & \textbf{0.934} & 35.27 & \textbf{0.966} \\
    DPS & 20.45 & 0.685 & 22.15 & 0.731 & 27.14 & 0.835 & 31.82 & 0.900 \\
    DAPS & 19.95 & 0.662 & 20.99 & 0.695 & 30.70 & 0.845 & 33.18 & 0.876 \\
    DDS & 21.00 & 0.703 & 24.54 & 0.769 & \underline{33.23} & 0.907 & \underline{37.51} & 0.937 \\
    \cdashline{1-9}\noalign{\vskip 3pt}
    Grad-GTP & \textbf{28.64} & \underline{0.828} & \underline{29.80} & \textbf{0.893} & 32.81 & \underline{0.919} & 36.62 & 0.914 \\
    Prox-GTP & \underline{28.58} & \textbf{0.843} & \textbf{30.57} & 0.856 & \textbf{35.15} & 0.917 & \textbf{37.73} & \underline{0.965} \\
    \bottomrule
    \end{tabular}
    \caption{Average PSNR and SSIM results attained by the proposed GTP methods and baselines for four LA-CT settings. In each column, the best result is marked in \textbf{bold}, while the second best is \underline{underlined}. For reference, the translation model used by GTP attains a PSNR of 24.94 dB and an SSIM of 0.839 without measurement guidance.}
    \label{tab:la_ct_main}
\end{figure*}

\subsection{CT Reconstruction with MRI Guidance}

We consider both \textit{sparse-view and limited-angle CT (SV/LA-CT)} reconstruction under the linear version of~\eqref{Eq:Inverse}, where $\Abm$ denotes the CT projection operator and $\ebm$ is additive Gaussian noise with an input \textit{signal-to-noise ratio (SNR)} of 50 dB.
For both settings, we adopt a parallel-beam geometry and implement the forward and adjoint operators using LEAP~\cite{kim2023differentiable}.
We use the combined 2023 and 2025 SynthRad datasets~\cite{synthrad2023,synthrad2025}, with volumes randomly split into training, validation, and test sets at an 80/10/10 ratio.
All volumes are resampled to an axial resolution of $256\times256$ using trilinear interpolation, with both CT and MR images normalized to $[-1,1]$.
We use five validation images for hyperparameter tuning and report final results on 30 images from the test set.

We compare GTP against baselines spanning three representative classes: conventional CT reconstruction methods, including \textit{filtered back-projection (FBP)} and \textit{total variation (TV)} regularization~\cite{sidky2008image}; cross-modality reconstruction using \textit{joint total variation (JTV)}~\cite{Ehrhardt.etal2016,sapiro1996anisotropic}, which incorporates the reference MR image through a hand-crafted joint regularizer; and state-of-the-art diffusion-based inverse solvers, including \textit{DPS}~\cite{Chung.etal2023diffusion}, \textit{DDS}~\cite{Chung.etal2024decomposed}, and \textit{DAPS}~\cite{zhang2025improving}.
For the diffusion baselines, we use DDBMs trained with a standard Gaussian source as unconditional image priors. These models employ the same architecture, dataset, and training procedure as the translation model used by our GTP algorithms.
All methods are evaluated using \textit{peak signal-to-noise ratio (PSNR)} and \textit{structural similarity index measure (SSIM)}~\cite{Wang.etal2004}. 
For each reconstruction method, relevant hyperparameters are selected by grid search to maximize the average PSNR over the validation images.
We refer to Supplement C.1 and C.2 for additional implementation and data preprocessing details.

\paragraph{GTP Implementation}
For Grad-GTP, we use the time-varying likelihood in~\eqref{eq:dps_likelihood}.
To improve numerical stability, we normalize the guidance step size by the likelihood-gradient magnitude, \textit{i.e.}, proportional to $1/\|\nabla_{\xbm}g(\xbm,t)\|_2$.
We apply the same normalization to DPS for a controlled comparison.
For Prox-GTP, we adopt an efficient implementation that applies data consistency to the denoised estimate $\hat{\xbm}_0(\xbm,t)$ by solving
\begin{equation}
\label{eq:prox_gtp_default}
\xbm_\mathsf{dc}
=
\underset{\zbm}{\argmin}~
\frac{1}{2}\|\zbm-\hat{\xbm}_0\|_2^2
+
\lambda(1-t)\frac{1}{2\beta^2}
\|\Abm\zbm-\ybm\|_2^2.
\end{equation}
This quadratic proximal update can be efficiently solved using the conjugate gradient method, after which $\xbm_{\mathsf{dc}}$ is renoised through the forward diffusion process to obtain the next iterate.
We note that this approximation substantially reduces the computational cost and yields stronger performance.

We additionally consider a strict implementation of Prox-GTP using the same time-varying likelihood as Grad-GTP, whose proximal update is solved iteratively without the above approximation. GTP also offers empirical flexibility to incorporate algorithmic designs from existing diffusion inverse solvers. 
We demonstrate this flexibility through DAPS-GTP and DDS-GTP, which replace the unconditional diffusion priors in DAPS and DDS with the same translation model used by GTP. These variants introduce additional approximations and therefore fall outside the direct scope of our theoretical analysis. 
Implementation details and quantitative results for all these variants are provided in Supplement C.3 and D.2--D.3.

\begin{figure*}[t!]
    \centering
    \includegraphics[width=0.95\linewidth]{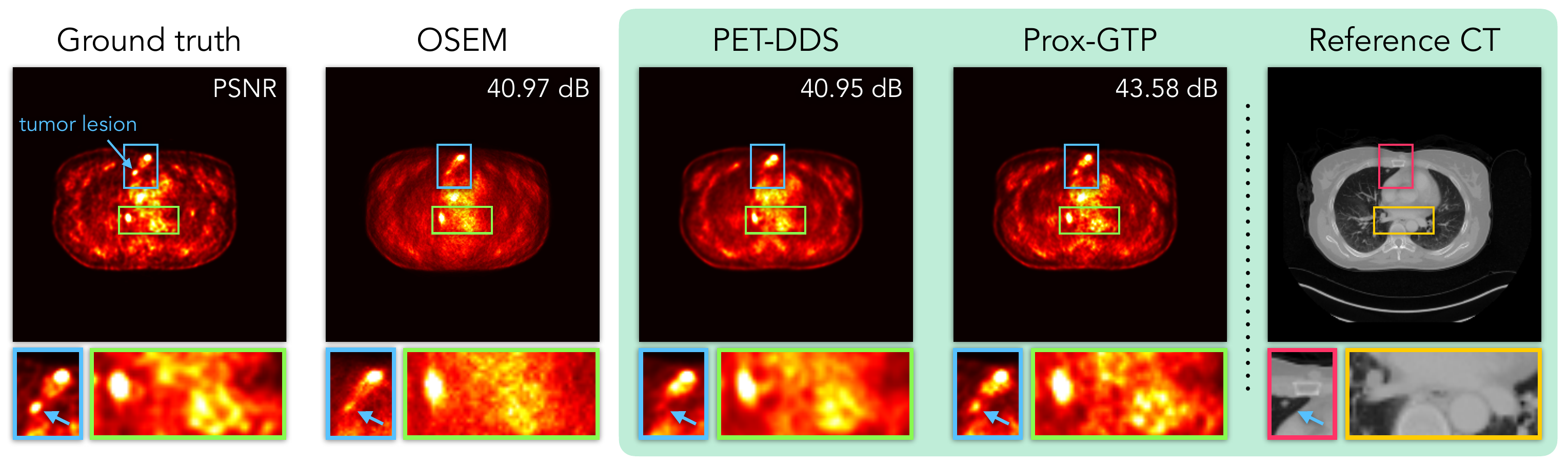}
    \caption{Visual comparison of low-dose PET reconstructions with $5\times10^5$ total counts. The shaded panels indicate methods that leverage the reference CT image during reconstruction. Zoomed-in regions highlight differences in fine-detail recovery. Note how Prox-GTP more clearly reconstructs the small tumor lesion, which is substantially less visible in the baseline reconstructions.}
    \label{fig:pet}
\end{figure*}

\paragraph{Sparse-View CT Results}

\Cref{tab:sv_ct_main} summarizes the average PSNR and SSIM for SV-CT reconstruction with $\{4, 8, 16, 32\}$ projection views, with projection angles equispaced over $[0^\circ,180^\circ)$.
The GTP methods perform particularly well in the more sparsely sampled settings.
At 4 views, Prox-GTP achieves the highest PSNR and SSIM, with a PSNR improvement of more than 2 dB over DDS, the strongest baseline without side-modality information.
As more projections become available, this performance margin gradually decreases, with GTP and the leading baselines achieving comparable performance at 16 and 32 views.
This trend, however, is consistent with the role of the translation prior: when the CT measurements are severely undersampled, the corresponding MR image provides substantial complementary information about the underlying anatomy; as the measurements become more informative, the additional benefit of this side information naturally diminishes.
\Cref{fig:sv_ct} provides a complementary visual comparison. 
While the unconditional methods generate plausible CT images, several anatomical structures around the kidneys are inconsistent with the ground truth, whereas the GTP methods better preserve the subject-specific anatomy available from the corresponding MR image.

Furthermore, the results also highlight the importance of how side-modality information is incorporated. 
While JTV also uses the corresponding MR image and achieves competitive PSNR and SSIM, its reconstructions in~\Cref{fig:sv_ct} are noticeably smoother and contain less anatomical detail. 
In contrast, GTP more faithfully recovers structures such as the kidney anatomy through a learned generative model of the cross-modal relationship. 
This observation highlights the advantage of GTP over hand-crafted cross-modal regularization in capturing complex inter-modal dependencies.

\paragraph{Limited-Angle CT}
\Cref{tab:la_ct_main} summarizes the average PSNR and SSIM for LA-CT reconstruction, where we consider an angular range of $\{30^\circ, 60^\circ, 90^\circ, 120^\circ\}$.
Projections are acquired at $1^\circ$ increments over $[0^\circ, n^\circ)$ for an $n$-degree angular range.
The GTP methods achieve the highest PSNR across all four settings, with their advantage becoming more pronounced as the angular coverage decreases.
For the most challenging $30^\circ$ setting, Prox-GTP outperforms all diffusion-based baselines by more than 7 dB in PSNR.
The improvement remains substantial at $90^\circ$, where Prox-GTP exceeds DDS, the closest competing method, by 1.92 dB.
In addition to the trend observed in SV-CT, LA-CT particularly benefits from cross-modal information because a contiguous range of projection angles is entirely unobserved, leaving the corresponding structural information inaccessible from the CT measurements alone.

\Cref{fig:la_ct} provides a complementary visual comparison for the $90^\circ$ setting.
Methods without side-modality information fail to recover major anatomical features that are poorly constrained by the acquired angular range, with particularly pronounced errors in bone structures.
In contrast, the GTP methods exploit the corresponding MR image to recover these structures with substantially higher anatomical fidelity.
JTV similarly benefits from access to the MR image, but its hand-crafted cross-modal regularization remains limited in recovering fine structures within the missing-angle regions.
The comparison further illustrates the advantage of GTP for extracting complementary anatomical information when the target measurements leave substantial structural information unobserved. 

\subsection{PET Reconstruction with CT Guidance}

Finally, we consider low-dose PET reconstruction with CT guidance, a realistic cross-modality setting where the two modalities are routinely acquired together and provide complementary functional and anatomical information.
Moreover, PET measurements follow Poisson statistics, allowing us to evaluate the applicability of GTP beyond Gaussian inverse problems.

We model the PET measurements as
$\ybm\sim\mathsf{Pois}(\Abm\xbm)$, where $\Abm$ denotes the PET forward operator.
Attenuation correction is incorporated into $\Abm$ using the corresponding CT image, following the standard piecewise-linear mapping from CT \textit{Hounsfield units (HU)} to attenuation coefficients~\cite{kinahan1998attenuation,shreve2011clinical}.
We implement the PET forward and adjoint operators using \texttt{parallelproj}~\cite{schramm2024parallelproj} and simulate a 2D scanner with a 12-sided polygonal detector ring, roughly approximating the geometry of the Siemens Biograph mCT scanner used to acquire the original data. We use the FDG-PET/CT Lesions dataset~\cite{gatidis2022whole}, which contains 900 paired PET/CT volumes with manual tumor-lesion segmentations, divided into 800/50/50 volumes for training, validation, and testing. All volumes are cropped to an axial resolution of $256\times256$; CT images are normalized to $[-1,1]$, while PET images are preprocessed following~\cite{singh2024score} to map most intensities to the same range. We use three randomly selected validation slices for hyperparameter tuning and evaluate on 50 test slices, including 30 slices containing tumor lesions and 20 randomly selected slices.

We consider two representative baselines for comparison: \textit{ordered subsets expectation maximization (OSEM)}~\cite{osem} and \textit{PET-DDS}~\cite{singh2024score}. 
The former represents a classical iterative PET reconstruction method, while the latter provides a diffusion-based comparison that also incorporates anatomical side information. Unlike GTP, which directly models the generative translation from CT to PET, PET-DDS uses the anatomical image as a conditioning input to a noise-to-data diffusion process. 
For our implementation of PET-DDS, we instantiate this conditional prior using a DDBM conditioned on the CT image with a standard Gaussian source distribution. 
We refer to Supplement C.4 for additional implementation details.

\begin{table}[t!]
    \centering
    \small
    \begin{tabular}{lcccccc}
    \toprule
    \multirow{2.4}{*}{Method\quad} & \multicolumn{2}{c}{$10^5$ counts} & \multicolumn{2}{c}{$5\times10^5$ counts} & \multicolumn{2}{c}{$10^6$ counts} \\ \cmidrule(lr){2-3} \cmidrule(lr){4-5} \cmidrule(lr){6-7}
     & PSNR & SSIM & PSNR & SSIM & PSNR & SSIM \\
    \midrule
    OSEM & 32.15 & 0.890 & 35.11 & 0.926 & 35.66 & 0.936 \\
    PET-DDS & \textbf{34.74} & \textbf{0.935} & \underline{36.24} & \underline{0.946} & \underline{36.42} & \underline{0.948} \\
    Prox-GTP & \underline{34.34} & \underline{0.934} & \textbf{37.34} & \textbf{0.951} & \textbf{38.66} & \textbf{0.959} \\
    \bottomrule
    \end{tabular}
    \caption{Average PSNR and SSIM results attained by the proposed GTP methods and baselines for three PET reconstruction settings. In each column, the best result is marked in \textbf{bold}, while the second best is \underline{underlined}. For reference, the translation model used by GTP attains a PSNR of 29.80 dB and an SSIM of 0.891 without measurement guidance.}
    \label{tab:quant_comp_pet}
\end{table}

\paragraph{Prox-GTP Implementation}
We focus on Prox-GTP for the PET experiments given its superior empirical performance over Grad-GTP in the preceding CT experiments.
We follow the same denoise--data-consistency--renoise structure as in the Gaussian settings.
At each sampling step, we compute the denoised estimate $\hat{\xbm}_0$ and apply data consistency by solving
\begin{equation}
\label{eq:pet_dds_prox}
\xbm_\mathsf{dc}
=
\underset{\zbm}{\argmin}~
\frac{1}{2}\|\zbm-\hat{\xbm}_0\|_2^2
+
\lambda(1-t)g(\zbm),
\end{equation}
where $g(\zbm)$ is given by
\begin{equation*}
g(\zbm)=\sum_i\left[(\Abm\zbm)_{(i)}-\ybm_{(i)}\log(\Abm\zbm)_{(i)}\right].
\end{equation*}
Here, index $(i)$ denotes the $i$th element of the vector.
As the Poisson likelihood can be expressed as a \textit{Bregman divergence}, we solve the nonquadratic proximal problem using Bregman proximal gradient descent. This approach exploits the structure of the Poisson likelihood and reduces each iteration to a gradient step involving $\Abm$ and $\Abm^T$, followed by a coordinate-wise Bregman proximal update. The resulting $\xbm_{\mathsf{dc}}$ is then renoised to the appropriate diffusion level.
We refer to Supplement C.5 for additional implementation details.

\paragraph{Reconstruction Results}
\Cref{tab:quant_comp_pet} summarizes the average PSNR and SSIM for PET reconstruction at three dose levels, corresponding to total expected counts of $10^5$, $5\times10^5$, and $10^6$. Both diffusion-based methods consistently outperform the classical OSEM baseline, with the improvement becoming particularly pronounced at lower count levels. Among the diffusion-based methods, Prox-GTP achieves the highest PSNR and SSIM at both $5\times10^5$ and $10^6$ counts, improving upon PET-DDS by 1.10 dB and 2.24 dB in PSNR, respectively. At the very-low-dose setting of $10^5$ counts, the two methods perform comparably, with PET-DDS exceeding Prox-GTP by only 0.40 dB in PSNR and 0.001 in SSIM. Moreover, compared with the translation model alone (29.80 dB), incorporating the PET measurements through GTP yields substantial improvements even at the lowest count level, highlighting the complementary roles of the cross-modal prior and measurement guidance.

\Cref{fig:pet} provides a complementary visual comparison at $5\times10^5$ counts, revealing differences in local structures that are less apparent from the global image-quality metrics. In particular, a small tumor lesion appears as a distinct hot spot in the ground-truth PET image. Prox-GTP faithfully recovers this lesion, whereas it is substantially less pronounced in the OSEM and PET-DDS reconstructions.

\section{Conclusion}

In this work, we introduced GTP, a principled framework for using generative translation models as expressive image priors for cross-modal image reconstruction. 
GTP achieves this by incorporating a time-varying measurement likelihood into the sampling process. 
We presented an analysis of how this alteration to the particle-level dynamics affects the evolution of the resulting probability density, revealing that injecting the likelihood in this manner leads to an intrinsic bias in the sampling procedure. 
We further show that this bias can be estimated without access to ground-truth information, yielding a reference-free criterion for assessing the accuracy of the recovered posterior.
We introduced two practical GTP algorithms based on different discretizations of the likelihood-adapted dynamics, Grad-GTP and Prox-GTP. We validated these proposed algorithms on multiple challenging cross-modality imaging tasks, including limited-angle and sparse-view CT reconstruction with MRI side information and low-dose PET reconstruction with CT side information. 
The empirical results demonstrate the benefits of the proposed approach over both conventional reconstruction methods and diffusion-based inverse solvers. 

\section*{Acknowledgments}
Y. Sun acknowledges support from the U.S. National Science Foundation (NSF) under Grant CCF-2542022. E. Bell is supported by the U.S. Department of
Energy, Office of Science, Office of Advanced Scientific Computing Research, Department of
Energy Computational Science Graduate Fellowship under Award No. DE-SC0026073. Y. Chen acknowledges support from the AMS-Simons Travel Grant and National Science Foundation award DMS-2608264.

\bibliographystyle{IEEEtran}
\bibliography{references}



\section*{Supplementary material}

# Supplementary Material for "Generative Translation Priors: Bayesian Imaging with Cross-Modality Image Translation"

**Evan Bell**[1], **Jiaming Liu**[2], **Yifan Chen**[3], **and Yu Sun**[1]

[1]Johns Hopkins University [2]Stanford University [3]University of California, Los Angeles

## Contents



# A Technical Proofs

**Notations.** Throughout the proofs, we consider the probability space $(\Omega, \mathscr{F}, \mathbb{P})$, where $\Omega$ denotes the sample space, $\mathscr{F}$ the $\sigma$-algebra, and $\mathbb{P}$ the probability measure. For a random variable $\zeta : \Omega \to \mathbb{R}^n$, we denote its expectation by $\mathbb{E}[\zeta] = \int_\Omega \zeta(\omega)\mathbb{P}(d\omega)$. Unless otherwise stated, all probability densities considered in the proofs are assumed to be sufficiently regular so that the quantities below are well defined. In particular, we assume that $g(\boldsymbol{x}, t)$ is continuously differentiable with respect to both $\boldsymbol{x}$ and $t$. We use the following expression of the *Kullback–Leibler (KL)* divergence for two probability densities $\nu$ and $\pi$

$$\mathsf{KL}\left(\nu \,\|\, \pi\right) = \int_{\mathbb{R}^n} \nu(\boldsymbol{x}) \log \frac{\nu(\boldsymbol{x})}{\pi(\boldsymbol{x})}\, d\boldsymbol{x} = \int \nu \log \frac{\nu}{\pi}.$$

The corresponding *relative Fisher information (FI)* is defined as

$$\mathsf{FI}\left(\nu \,\|\, \pi\right) = \int_{\mathbb{R}^n} \nu(\boldsymbol{x}) \left\| \nabla \log \frac{\nu(\boldsymbol{x})}{\pi(\boldsymbol{x})} \right\|_2^2 \, d\boldsymbol{x} = \int \nu \left\| \nabla \log \frac{\nu}{\pi} \right\|_2^2 .$$

## A.1 Proof of Proposition 1

By construction, we can derive the evolution PDE that $q_t$ follows

$$\begin{aligned}
\frac{\partial}{\partial t} q_t(\boldsymbol{x}) &= e^{-g(\boldsymbol{x},t)} \frac{\partial}{\partial t} p_t(\boldsymbol{x}) - e^{-g(\boldsymbol{x},t)} \left( \frac{\partial}{\partial t} g(\boldsymbol{x},t) \right) p_t(\boldsymbol{x}) \\
&= e^{-g(\boldsymbol{x},t)} \left( \frac{\partial}{\partial t} p_t(\boldsymbol{x}) - \left( \frac{\partial}{\partial t} g(\boldsymbol{x},t) \right) p_t(\boldsymbol{x}) \right) \\
&= e^{-g(\boldsymbol{x},t)} \left( -\mathsf{div}\Big( b(\boldsymbol{x},t)\, p_t(\boldsymbol{x}) \Big) + \frac{\sigma_t^2}{2} \Delta p_t(\boldsymbol{x}) - \left( \frac{\partial}{\partial t} g(\boldsymbol{x},t) \right) p_t(\boldsymbol{x}) \right). \qquad (1)
\end{aligned}$$

To eliminate $p_t(\boldsymbol{x})$, we apply the change of variable $p_t = e^{g(\boldsymbol{x},t)} q_t(\boldsymbol{x})$ and expand the derivatives

$$\frac{\partial}{\partial t} q_t(\boldsymbol{x}) = e^{-g(\boldsymbol{x},t)} \left( -\mathsf{div}\Big( b(\boldsymbol{x},t)\, e^{g(\boldsymbol{x},t)} q_t(\boldsymbol{x}) \Big) + \frac{\sigma_t^2}{2} \Delta e^{g(\boldsymbol{x},t)} q_t(\boldsymbol{x}) - \left( \frac{\partial}{\partial t} g(\boldsymbol{x},t) \right) e^{g(\boldsymbol{x},t)} q_t(\boldsymbol{x}) \right). \qquad (2)$$

Two useful equalities are as follows

$$\begin{aligned}
\Delta e^{g(\boldsymbol{x},t)} q_t(\boldsymbol{x}) &= e^{g(\boldsymbol{x},t)} \Big( \Delta q_t(\boldsymbol{x}) + 2 \nabla q_t(\boldsymbol{x}) \cdot \nabla g(\boldsymbol{x},t) + q_t(\boldsymbol{x}) \Delta g(\boldsymbol{x},t) + q_t(\boldsymbol{x}) \|\nabla g(\boldsymbol{x},t)\|_2^2 \Big) \\
-\,\mathsf{div}\Big( b(\boldsymbol{x},t)\, e^{g(\boldsymbol{x},t)} q_t(\boldsymbol{x}) \Big) &= -e^{g(\boldsymbol{x},t)} \Big( \mathsf{div}\Big( b(\boldsymbol{x},t) q_t(\boldsymbol{x}) \Big) + b(\boldsymbol{x},t) \cdot \nabla g(\boldsymbol{x},t) q_t(\boldsymbol{x}) \Big).
\end{aligned}$$

Plug in the two equalities and simplify the terms

$$\begin{aligned}
\frac{\partial}{\partial t} q_t(\boldsymbol{x}) =& e^{-g(\boldsymbol{x},t)} \left( -\mathsf{div}\Big( b(\boldsymbol{x},t)\, e^{g(\boldsymbol{x},t)} q_t(\boldsymbol{x}) \Big) + \frac{\sigma_t^2}{2} \Delta e^{g(\boldsymbol{x},t)} q_t(\boldsymbol{x}) - \left( \frac{\partial}{\partial t} g(\boldsymbol{x},t) \right) e^{g(\boldsymbol{x},t)} q_t(\boldsymbol{x}) \right) \\
=& -\mathsf{div}\Big( b(\boldsymbol{x},t) q_t(\boldsymbol{x}) \Big) + \sigma_t^2 \nabla q_t(\boldsymbol{x}) \cdot \nabla g(\boldsymbol{x},t) + \frac{\sigma_t^2}{2} \Delta q_t(\boldsymbol{x}) \\
&+ q_t(\boldsymbol{x}) \left( -\frac{\partial}{\partial t} g(\boldsymbol{x},t) - b(\boldsymbol{x},t) \cdot \nabla g(\boldsymbol{x},t) + \frac{\sigma_t^2}{2} \Delta g(\boldsymbol{x},t) + \frac{\sigma_t^2}{2} \|\nabla g(\boldsymbol{x},t)\|_2^2 \right) \\
=& \underbrace{-\mathsf{div}\Big( b(\boldsymbol{x},t) q_t(\boldsymbol{x}) \Big) + \sigma_t^2 \nabla q_t(\boldsymbol{x}) \cdot \nabla g(\boldsymbol{x},t) + q_t(\boldsymbol{x}) \sigma_t^2 \Delta g(\boldsymbol{x},t)}_{= -\mathsf{div}\left( (b(\boldsymbol{x},t) - \sigma_t^2 \nabla g(\boldsymbol{x},t)) q_t(\boldsymbol{x}) \right)} + \frac{\sigma_t^2}{2} \Delta q_t(\boldsymbol{x}) \\
&+ q_t(\boldsymbol{x}) \left( -\frac{\partial}{\partial t} g(\boldsymbol{x},t) - b(\boldsymbol{x},t) \cdot \nabla g(\boldsymbol{x},t) - \frac{\sigma_t^2}{2} \Delta g(\boldsymbol{x},t) + \frac{\sigma_t^2}{2} \|\nabla g(\boldsymbol{x},t)\|_2^2 \right). \qquad (3)
\end{aligned}$$

By properly rearranging the terms, part of (3) can be written as the Fokker-Planck equation of the following SDE

$$d\boldsymbol{x}_t = \Big( b(\boldsymbol{x}_t,t) - \sigma_t^2 \, \nabla g(\boldsymbol{x}_t,t) \Big) dt + \sigma_t d\boldsymbol{B}_t.$$

We can further translate (3) for $\pi_t$ by adding back the normalizing constant

$$\frac{\partial}{\partial t}\pi_t(\boldsymbol{x}) = \frac{1}{Z_t}\frac{\partial}{\partial t}q_t(\boldsymbol{x}) - \frac{\pi_t(\boldsymbol{x})}{Z_t}\frac{\partial}{\partial t}Z_t = \frac{1}{Z_t}\frac{\partial}{\partial t}q_t(\boldsymbol{x}) - \pi_t(\boldsymbol{x})\frac{\partial}{\partial t}\log Z_t.$$

Hence, we can achieve

$$\begin{aligned}\frac{\partial}{\partial t}\pi_t(\boldsymbol{x}) &= -\mathrm{div}\big((b(\boldsymbol{x},t) - \sigma_t^2\nabla g(\boldsymbol{x},t))\pi_t(\boldsymbol{x})\big) + \frac{\sigma_t^2}{2}\Delta\pi_t(\boldsymbol{x}) \\ &\quad + \pi_t(\boldsymbol{x})\left(-\frac{\partial}{\partial t}g(\boldsymbol{x},t) - b(\boldsymbol{x},t)\cdot\nabla g(\boldsymbol{x},t) - \frac{\sigma_t^2}{2}\Delta g(\boldsymbol{x},t) + \frac{\sigma_t^2}{2}\|\nabla g(\boldsymbol{x},t)\|_2^2 - \frac{\partial}{\partial t}\log Z_t\right), \quad (4)\end{aligned}$$

which completes the proof.

## A.2 Proof of Proposition 2

First write the Fokker-Planck equation associated with the SDE

$$d\boldsymbol{x}_t = \Big(b(\boldsymbol{x}_t,t) - \alpha_t\,\nabla g(\boldsymbol{x}_t,t)\Big)\,dt + \sigma_t\,d\boldsymbol{B}_t. \quad (5)$$

We use $\widehat{\pi}_t$ to denote the time-dependent probability

$$\begin{aligned}\frac{\partial}{\partial t}\widehat{\pi}_t(\boldsymbol{x}) &= -\mathrm{div}\Big(\big(b(\boldsymbol{x},t) - \alpha_t\nabla g(\boldsymbol{x},t)\big)\,\widehat{\pi}_t(\boldsymbol{x})\Big) + \frac{\sigma_t^2}{2}\Delta\widehat{\pi}_t(\boldsymbol{x}) \\ &= -\mathrm{div}\Big(b(\boldsymbol{x},t)\widehat{\pi}_t(\boldsymbol{x})\Big) + \alpha_t\nabla\widehat{\pi}_t(\boldsymbol{x})\cdot\nabla g(\boldsymbol{x},t) + \alpha_t\widehat{\pi}_t(\boldsymbol{x})\Delta g(\boldsymbol{x},t) + \frac{\sigma_t^2}{2}\Delta\widehat{\pi}_t(\boldsymbol{x}) \\ &= -\mathrm{div}\Big(b(\boldsymbol{x},t)\widehat{\pi}_t(\boldsymbol{x})\Big) + \alpha_t\widehat{\pi}_t(\boldsymbol{x})\Big(\nabla\log\widehat{\pi}_t(\boldsymbol{x})\cdot\nabla g(\boldsymbol{x},t)\Big) + \frac{\sigma_t^2}{2}\Delta\widehat{\pi}_t(\boldsymbol{x}) \\ &\quad + \alpha_t\widehat{\pi}_t(\boldsymbol{x})\Delta g(\boldsymbol{x},t), \quad (6)\end{aligned}$$

where we use the fact

$$\nabla\widehat{\pi}_t(\boldsymbol{x}) = \widehat{\pi}_t\nabla\log\widehat{\pi}_t(\boldsymbol{x}) \quad\Rightarrow\quad \nabla\widehat{\pi}_t(\boldsymbol{x})\cdot\nabla g(\boldsymbol{x},t) = \widehat{\pi}_t(\boldsymbol{x})\Big(\nabla\log\widehat{\pi}_t(\boldsymbol{x})\cdot\nabla g(\boldsymbol{x},t)\Big).$$

Similarly, we can expand (3)

$$\begin{aligned}\frac{\partial}{\partial t}q_t(\boldsymbol{x}) &= -\,\mathrm{div}\Big(b(\boldsymbol{x},t)q_t(\boldsymbol{x})\Big) + \sigma_t^2 q_t(\boldsymbol{x})\Big(\nabla\log q_t(\boldsymbol{x})\cdot\nabla g(\boldsymbol{x},t)\Big) + \frac{\sigma_t^2}{2}\Delta q_t(\boldsymbol{x}) \\ &\quad + q_t(\boldsymbol{x})\left(-\frac{\partial}{\partial t}g(\boldsymbol{x},t) - b(\boldsymbol{x},t)\cdot\nabla g(\boldsymbol{x},t) + \frac{\sigma_t^2}{2}\Delta g(\boldsymbol{x},t) + \frac{\sigma_t^2}{2}\|\nabla g(\boldsymbol{x},t)\|_2^2\right) \quad (7)\end{aligned}$$

Comparing the structures of (7) and (6) yields

$$\begin{aligned}\frac{\partial}{\partial t}q_t(\boldsymbol{x}) &= \mathrm{FPE}_{(6)}(q_t(\boldsymbol{x})) + q_t(\boldsymbol{x})\Bigg(-\frac{\partial}{\partial t}g(\boldsymbol{x},t) - b(\boldsymbol{x},t)\cdot\nabla g(\boldsymbol{x},t) + \big(\sigma_t^2 - \alpha_t\big)\,\nabla\log q_t(\boldsymbol{x})\cdot\nabla g(\boldsymbol{x},t) \\ &\quad + \left(\frac{\sigma_t^2}{2} - \alpha_t\right)\Delta g(\boldsymbol{x},t) + \frac{\sigma_t^2}{2}\|\nabla g(\boldsymbol{x},t)\|_2^2\Bigg)\end{aligned}$$

$$= \mathsf{FPE}_{(6)}(q_t(\boldsymbol{x})) + q_t(\boldsymbol{x})\Bigg(-\frac{\partial}{\partial t}g(\boldsymbol{x},t) - b(\boldsymbol{x},t)\cdot\nabla g(\boldsymbol{x},t) + \left(\sigma_t^2-\alpha_t\right)\nabla\log p_t(\boldsymbol{x})\cdot\nabla g(\boldsymbol{x},t)$$

$$-\left(\sigma_t^2-\alpha_t\right)\|\nabla g(\boldsymbol{x},t)\|_2^2 + \left(\frac{\sigma_t^2}{2}-\alpha_t\right)\Delta g(\boldsymbol{x},t) + \frac{\sigma_t^2}{2}\|\nabla g(\boldsymbol{x},t)\|_2^2\Bigg)$$

$$= \mathsf{FPE}_{(6)}(q_t(\boldsymbol{x})) + q_t(\boldsymbol{x})\Bigg(-\frac{\partial}{\partial t}g(\boldsymbol{x},t) - b(\boldsymbol{x},t)\cdot\nabla g(\boldsymbol{x},t) + \left(\sigma_t^2-\alpha_t\right)\nabla\log p_t(\boldsymbol{x})\cdot\nabla g(\boldsymbol{x},t)$$

$$+\left(\frac{\sigma_t^2}{2}-\alpha_t\right)\Delta g(\boldsymbol{x},t) + \left(\alpha_t-\frac{\sigma_t^2}{2}\right)\|\nabla g(\boldsymbol{x},t)\|_2^2\Bigg), \tag{8}$$

where we use $\nabla\log q_t(\boldsymbol{x}) = -\nabla g(\boldsymbol{x},t) + \nabla\log p_t(\boldsymbol{x})$ in the second equality. Similarly, we can translate (8) for $\pi_t$ by adding back the normalizing constant, which yields

$$\frac{\partial}{\partial t}\pi_t(\boldsymbol{x}) = \mathsf{FPE}_{(6)}(\pi_t(\boldsymbol{x})) + \pi_t(\boldsymbol{x})\Bigg(-\frac{\partial}{\partial t}g(\boldsymbol{x},t) - b(\boldsymbol{x},t)\cdot\nabla g(\boldsymbol{x},t) + \left(\sigma_t^2-\alpha_t\right)\nabla\log p_t(\boldsymbol{x})\cdot\nabla g(\boldsymbol{x},t)$$

$$+\left(\frac{\sigma_t^2}{2}-\alpha_t\right)\Delta g(\boldsymbol{x},t) + \left(\alpha_t-\frac{\sigma_t^2}{2}\right)\|\nabla g(\boldsymbol{x},t)\|_2^2 - \frac{\partial}{\partial t}\log Z_t\Bigg). \tag{9}$$

It is clear that, by choosing $\alpha_t = \sigma_t^2/2$, we can simplify the computation of the bias by removing the Laplacian and squared-norm terms. This completes the proof.

## A.3 Proof of Theorem 1

We first characterize the time evolution of $\mathsf{KL}(\widehat{\pi}_t \,\|\, \pi_t)$. Related results without the intrinsic bias term have been established in Lemma A.3.2 of [1], and a more general formulation is provided in Lemma A.4 of [2]. To start, we write

$$\begin{aligned}
\frac{\partial}{\partial t}\mathsf{KL}(\widehat{\pi}_t \,\|\, \pi_t) &= \int \frac{\partial}{\partial t}\widehat{\pi}_t \log\frac{\widehat{\pi}}{\pi_t} + \frac{\partial}{\partial t}\left(\frac{\widehat{\pi}_t}{\pi_t}\right)\pi_t \\
&= \int \frac{\partial}{\partial t}\widehat{\pi}_t \log\frac{\widehat{\pi}}{\pi_t} + \underbrace{\frac{\partial}{\partial t}\int\widehat{\pi}_t}_{=0} - \int\left(\frac{\partial}{\partial t}\pi_t\right)\frac{\widehat{\pi}_t}{\pi_t} \\
&= \int\left[-\mathsf{div}\Big(\left(b_t-\alpha_t\nabla g_t\right)\widehat{\pi}_t\Big) + \frac{\sigma_t^2}{2}\Delta\widehat{\pi}_t\right]\log\frac{\widehat{\pi}_t}{\pi_t} \\
&\quad - \int\left[-\mathsf{div}\Big(\left(b_t-\alpha_t\nabla g_t\right)\pi_t\Big) + \frac{\sigma_t^2}{2}\Delta\pi_t\right]\frac{\widehat{\pi}_t}{\pi_t} - \int\widehat{\pi}_t\,\mathsf{bias}(\boldsymbol{x},t,\alpha_t)
\end{aligned} \tag{10}$$

where we use $b_t = b(\boldsymbol{x},t)$ and $\nabla g_t = \nabla g(\boldsymbol{x},t)$ for brevity. As shown in [1], the transport terms will cancel each other, *i.e.*

$$\int\mathsf{div}\Big(\left(b_t-\alpha_t\nabla g_t\right)\widehat{\pi}_t\Big)\log\frac{\widehat{\pi}_t}{\pi_t} = -\int\left\langle (b_t-\alpha_t\nabla g_t)\widehat{\pi}_t,\ \nabla\log\frac{\widehat{\pi}_t}{\pi_t}\right\rangle = \int\mathsf{div}\Big((b_t-\alpha_t\nabla g_t)\pi_t\Big)\frac{\widehat{\pi}_t}{\pi_t}. \tag{11}$$

Combining the two results yields

$$\frac{\partial}{\partial t}\mathsf{KL}(\widehat{\pi}_t \,\|\, \pi_t) = \int\frac{\sigma_t^2}{2}\left((\Delta\widehat{\pi}_t)\log\frac{\widehat{\pi}_t}{\pi_t} - (\Delta\pi_t)\frac{\widehat{\pi}_t}{\pi_t}\right) - \int\widehat{\pi}_t\,\mathsf{bias}(\boldsymbol{x},t,\alpha_t)$$

$$
\begin{aligned}
&= -\frac{\sigma_t^2}{2}\int \langle \nabla\widehat{\pi}_t,\, \nabla\log\frac{\widehat{\pi}_t}{\pi_t}\rangle - \langle \nabla\pi_t,\, \nabla\frac{\widehat{\pi}_t}{\pi_t}\rangle - \int \widehat{\pi}_t\,\mathsf{bias}(\boldsymbol{x},t,\alpha_t)\\
&= -\frac{\sigma_t^2}{2}\int \langle \nabla\widehat{\pi}_t - \nabla\pi_t\frac{\widehat{\pi}_t}{\pi_t},\, \nabla\log\frac{\widehat{\pi}_t}{\pi_t}\rangle - \int \widehat{\pi}_t\,\mathsf{bias}(\boldsymbol{x},t,\alpha_t)\\
&= -\frac{\sigma_t^2}{2}\int \widehat{\pi}_t\left\|\nabla\log\frac{\widehat{\pi}_t}{\pi_t}\right\|_2^2 - \int \widehat{\pi}_t\,\mathsf{bias}(\boldsymbol{x},t,\alpha_t)\\
&= -\frac{\sigma_t^2}{2}\mathsf{FI}(\widehat{\pi}_t\,\|\,\pi_t) - \mathbb{E}_{\widehat{\pi}_t}\big[\mathsf{bias}(\boldsymbol{x},t,\alpha_t)\big],
\end{aligned}
\tag{12}
$$

where we use the fact

$$
\nabla\log\frac{\widehat{\pi}_t}{\pi_t} = \frac{\pi_t}{\widehat{\pi}_t}\nabla\frac{\widehat{\pi}_t}{\pi_t} = \frac{\nabla\widehat{\pi}_t}{\widehat{\pi}_t} - \frac{\nabla\pi_t}{\pi_t}.
$$

Integrating (12) over $t \in [0,T]$, dropping the nonpositive Fisher information term, and taking absolute values yields the result

$$
\mathsf{KL}(\widehat{\pi}_T\,\|\,\pi) \le \left|\int_0^T \mathbb{E}_{\widehat{\pi}_t}\big[\mathsf{bias}(\boldsymbol{x},t,\alpha_t)\big]\,dt\right|.
$$

The absolute value ensures that the resulting quantity serves as a nonnegative measure of accumulated bias.

# B Validation of bias theory for Grad-GTP

## B.1 Numerical experiments

We now provide a demonstration of the usefulness of the intrinsic bias criterion for Grad-GTP, similar to the corresponding experiment for Prox-GTP provided in Section IV. We use the same setup for generating the ground truth data and observed measurement as we did for Prox-GTP. We also use the same likelihood $g(\boldsymbol{x},t) = (1-t)\cdot\frac{1}{2\sigma^2}\|\boldsymbol{A}\hat{\boldsymbol{x}}_0(\boldsymbol{x},t)-\boldsymbol{y}\|_2^2$. We drew 200 samples using 200 discretization steps. The guidance strength is determined by a parameter $\lambda$, which multiplies the step size. We also adaptively rescale the gradient step size by dividing the step size by $\|\nabla g\|_2$. This is equivalent to setting the diffusion coefficient $\sigma_t^2 = 2\lambda/(\gamma\|\nabla g(\boldsymbol{x}_t,t)\|_2)$. We discuss the rationale for this design choice further in the next subsection.

The posteriors recovered by Grad-GTP for different guidance strengths $\lambda$ are shown in Figure 1, along with other quantitative measures of posterior fidelity. We plot the intrinsic bias, which upper bounds the KL divergence between the recovered and true posteriors, along with the (squared) maximum mean discrepancy and the negative log-likelihood of the ground truth image under the recovered posterior. We find that the intrinsic bias, MMD, and NLL exhibit the same trends and attain their respective minima at the same value of $\lambda$. This validates that the intrinsic bias can serve as a useful criterion for assessing the quality of the recovered posterior without access to the ground truth image or posterior distribution.

## B.2 Modifying the diffusion coefficient

In some cases, applying Algorithms 1 and 2 with the default schedule of the diffusion coefficient $\sigma_t$ can pose numerical challenges. In particular, for some choices of $g(\boldsymbol{x}_t,t)$, the default $\sigma_t$ may be too large. This issue could be alleviated by rescaling $g(\boldsymbol{x}_t,t)$ by a constant factor, but this may cause $q_t(\boldsymbol{x}) = e^{-g(\boldsymbol{x},t)}\,p_t(\boldsymbol{x})$ to no longer arrive at the true posterior. To handle this issue, in some cases we

modify the diffusion coefficient $\sigma_t$, which scales the step size on $\nabla g$. This approach is theoretically justified by the theory of stochastic interpolants [3], of which our DDBMs are a special case. In stochastic interpolants, users are permitted to specify an arbitrary schedule for $\sigma_t$ and still sample the same marginal distributions, provided that the drift $b(\boldsymbol{x}_t, t)$ is modified accordingly [3, Cor. 10, Cor. 18]. In practice, we keep the unmodified DDBM solver, although it is possible to manually modify the drift and use other numerical schemes. For a detailed discussion on how to interpret DDBMs in the framework of stochastic interpolants, we refer to [4].

### B.3 Computing the normalizing constant

Computing the intrinsic bias at each step requires estimating the change in the normalizing constant, which is determined by $\partial_t \log Z_t$. In order to accomplish this, we use the *harmonic mean estimator* of the normalizing constant. We briefly derive this estimator.

$$1 = \int p_t(\boldsymbol{x}) d\boldsymbol{x} = \int \frac{p_t(\boldsymbol{x})}{\pi_t(\boldsymbol{x})} \pi_t(\boldsymbol{x}) d\boldsymbol{x} = \int \frac{p_t(\boldsymbol{x})}{\frac{p_t(\boldsymbol{x}) e^{-g(\boldsymbol{x},t)}}{Z_t}} \pi_t(\boldsymbol{x}) d\boldsymbol{x} = Z_t \int \frac{1}{e^{-g(\boldsymbol{x},t)}} \pi_t(\boldsymbol{x}) d\boldsymbol{x} \tag{13}$$

$$\implies 1 = Z_t \mathbb{E}_{\boldsymbol{x} \sim \pi_t} \left[ \frac{1}{e^{-g(\boldsymbol{x},t)}} \right] \implies Z_t = \left( \mathbb{E}_{\boldsymbol{x} \sim \pi_t} \left[ \frac{1}{e^{-g(\boldsymbol{x},t)}} \right] \right)^{-1}. \tag{14}$$

The advantage of this estimator is that it is straightforward to compute using a batch of generated samples. On the other hand, a potential disadvantage of this approach is that it may suffer from high variance.

## C Implementation details for main experiments

### C.1 Training GTPs

We realize both the translation models used for our GTP algorithms and the diffusion models used for baseline methods as denoising diffusion bridge models (DDBMs) [5]. In particular, we use the variance preserving (VP) configuration of DDBMs. The design of VP-DDBMs is attractive for two reasons. First, VP-DDBMs are designed to translate samples from a source distribution to samples from a target distribution exactly between times $t = 1$ and $t = 0$. This makes the design choices of DDBMs naturally compatible with our theoretical framework. Second, DDBMs offer state-of-the-art performance on both generative image-to-image translation tasks and unconditional image generation (which can be performed by setting the source distribution to a standard Gaussian).

**Architecture and training.** We used the same model architecture and training settings for both the unconditional diffusion models and the translation models. The backbone is a standard diffusion U-Net adapted from [6]. We use a model with 64 base channels, channel multipliers of $[1, 2, 3, 4, 8]$, and attention at the $16 \times 16$ resolution. The model parameters were optimized using AdamW [7] with a learning rate of $10^{-4}$ with other optimizer hyperparameters set to their default values. The models were trained using four NVIDIA H200 GPUs for 250,000 steps using a per-GPU batch size of 16. During training, we maintained a per-step exponential moving average of the parameters with a decay rate of 0.9999. In all cases, the checkpoint used for evaluation is the averaged model from the end of training.

**Data preprocessing and selection for MR→CT.** For the MR-guided CT reconstruction experiments, we use the 2023 and 2025 SynthRad datasets [8, 9]. The 3D volumes were first preprocessed using TotalSegmentator [10] to remove image content outside the body, such as the scanner bed. The pixel values of the CT images in Hounsfield units (HU) were clamped to the range of $[-1000, 1000]$ and then linearly rescaled to $[-1, 1]$. The MR images were normalized by clamping their values between 0 and the $99.5$th percentile voxel value of each volume, then linearly rescaling the resulting values to the range $[-1, 1]$. The volumes were split into training, validation, and testing in an approximate 80/10/10 ratio. We randomly selected five images from the validation set to use for tuning the hyperparameters of each algorithm. For testing, we selected 30 slices from the testing dataset. The testing slices were chosen randomly from the middle 20 slices of random volumes (for reference, the test volumes contain about 60 slices on average). This was done to ensure that all testing slices were relatively free of artifacts and had a substantial amount of image content.

**Data preprocessing and selection for CT→PET.** Our CT-guided PET reconstruction experiments employ the FDG-PET/CT Lesions dataset [11]. The registered CT and PET volumes were cropped to have axial dimensions of $256 \times 256$. For the CT images, we followed the preprocessing routine outlined in the previous paragraph; we clamped the voxel values to the range $[-1000, 1000]$ HU and rescaled the resulting values to the range $[-1, 1]$. For PET images, preprocessing is more difficult, since PET images have a very large dynamic range which varies between subjects. To address this issue, we follow the normalization technique described in [12]. For each slice, we compute a constant $c$ as the total sum of activity in the image divided by the number of pixels that have any appreciable activity, that is for a PET image $\boldsymbol{x}$, we define

$$c = \frac{\sum_i \boldsymbol{x}_i}{\#\{k \mid \boldsymbol{x}_k > t\}}$$

where we set the threshold $t$ to be $1\%$ of the largest activity value of $\boldsymbol{x}$. The normalized image is then obtained by dividing $\boldsymbol{x}$ by $c$. This normalization strategy results in pixel values that are typically in the range of $[0, 1]$, although some "hot spots" may have larger values. For training, we multiply these normalized values by 2 and subtract 1, so most of the resulting pixel values are in the range $[-1, 1]$. The dataset consists of 900 volumes, which were divided into 800 volumes for training, 50 volumes for validation, and 50 volumes for testing. We randomly selected three slices from the validation set to use for hyperparameter tuning. For testing, we randomly selected 30 images that contain annotated tumor lesions. We then selected another 20 images at random from the testing volumes to reach a total of 50 images. To ensure that every test image is high quality, we filtered out images from the first and last ten slices of each volume (the test volumes have about 380 slices on average) and images with very little image content.

## C.2 Implementation details of baseline methods for CT reconstruction

**FBP.** The *filtered back-projection (FBP)* reconstruction was computed using LEAP [13], using the default filter settings.

**TV.** The *total variation (TV)* reconstruction was also performed using LEAP. We use the iterative adaptive steepest descent projection onto convex sets (ASD-POCS) algorithm [14] implemented in LEAP. The regularization term employs a smoothed TV penalty with the smoothing parameter $\delta$ set to $0.001$. For the reconstruction, we use 30 subsets and use 50 inner iterations to solve the TV denoising subproblem at each outer iteration.

**JTV.** For an image pair $\boldsymbol{u}_1$ and $\boldsymbol{u}_2$, the *joint TV (JTV)* regularizer is defined as $\mathsf{JTV}(\boldsymbol{u}_1, \boldsymbol{u}_2) = \int \sqrt{|\nabla \boldsymbol{u}_1(x)|^2 + |\nabla \boldsymbol{u}_2(x)|^2} dx$. This regularizer promotes $\boldsymbol{u}_1$ and $\boldsymbol{u}_2$ having edges at the same locations. Numerically, we implement JTV by adapting the implementation provided by [15], which can be found here[1]. This implementation uses the ASTRA toolbox [16, 17] to implement the X-ray projectors and solves the regularized reconstruction problem using the primal-dual hybrid gradient (PDHG) solver in ODL [18]. We use the default optimization parameters. This implementation includes two hyperparameters: the weight of the regularization term, $\alpha$, and a parameter $\eta$ that scales the values of the gradient of the side information ($\nabla \boldsymbol{u}_2$). We tuned these parameters for each task using a grid search over the sets $\alpha \in$ [1e-7, 5e-7, 1e-6, 5e-6, 1e-5, 5e-5, 1e-4, 5e-4, 1e-3] and $\eta \in$ [1e-3, 5e-3, 1e-2, 5e-2, 1e-1, 5e-1, 1.0, 5.0, 10.0].

**DPS.** We implement DPS [19] using our trained diffusion model. The primary hyperparameter in DPS is the choice of the step size for the likelihood guidance. In DPS, the step size is equal to $\lambda / \|\boldsymbol{A}\hat{\boldsymbol{x}}_0 - \boldsymbol{y}\|_2$, where $\lambda$ is a hyperparameter. We found that this choice performed very poorly when using a DDBM as the prior model. We choose to set the step size equal to $\lambda / \|\nabla_{\boldsymbol{x}} g(\boldsymbol{x}, t)\|_2$, where $g(\boldsymbol{x}, t) = \frac{1}{2\beta^2} \|\boldsymbol{A}\hat{\boldsymbol{x}}_0 - \boldsymbol{y}\|_2$. We tuned the hyperparameter $\lambda$ using a grid search over the set [1e-3, 2e-3, 5e-3, 1e-2, 2e-2, 5e-2, 0.1, 0.2, 0.5, 1.0, 2.0, 5.0, 10.0, 20.0]. For completeness, we also tested both choices of step size scaling (the original choice and our proposed choice), but found that our choice offered better performance on the validation set in all cases.

**DAPS.** We implemented DAPS [20] using our trained model by adapting the implementation from InverseBench [21]. DAPS alternates between two steps: (1.) reverse diffusion sampling to sample $\hat{\boldsymbol{x}}_0$ given $\boldsymbol{x}_t$ and (2.) using Langevin dynamics to sample $\boldsymbol{x}_0$ given $\hat{\boldsymbol{x}}_0$ and the measurement $\boldsymbol{y}$. In our implementation, step (1.) is accomplished using the DDBM solver and step (2.) is implemented following the existing code from [21]. We perform a grid search over the learning rate for the Langevin dynamics sampler over the set [1e-5, 5e-5, 1e-4] and the weight for the likelihood gradient over the set [1, 1e1, 1e2, 1e3, 1e4]. To sample $\hat{\boldsymbol{x}}_0$ given $\boldsymbol{x}_t$, we use 20 reverse diffusion steps at each iteration. The Langevin dynamics sampler is run for 50 steps at each outer iteration. We also found that the performance of DAPS was improved by integrating to a final time of $t$=1e-3 as opposed to the default of $t$=1e-4 in the DDBM solver, since the Langevin sampling steps became unstable for very small values of $t$. For DAPS we tuned the hyperparameters using only one validation image due to its relatively long runtime (about 3 minutes per image).

**DDS.** We implemented DDS using our trained diffusion model. The design of DDS is similar to our proposed default Prox-GTP. At each iteration, the denoised image is first computed as $\hat{\boldsymbol{x}}_0 = D_\theta(\boldsymbol{x}, t)$. Then, this estimate is updated using the acquired measurements by solving the optimization problem:

$$\boldsymbol{x}_{\mathsf{dc}} = \mathop{\mathsf{arg\,min}}_{\boldsymbol{z}} \frac{1}{2}\|\boldsymbol{z} - \hat{\boldsymbol{x}}_0\|_2^2 + \frac{\lambda}{2\beta^2}\|\boldsymbol{A}\boldsymbol{z} - \boldsymbol{y}\|_2^2.$$

The next iterate is then computed by renoising $\boldsymbol{x}_{\mathsf{dc}}$ to the appropriate noise level. This optimization problem can be solved efficiently using the conjugate gradient method. We tuned the parameters $\lambda$ and number of CG steps $n_{\mathsf{CG}}$ by grid search over the sets $\lambda \in$ [1e-5, 1e-4, 1e-3, 1e-2, 1e-1, 1.0, 1e1, 1e2, 1e3, 1e4] and $n_{\mathsf{CG}} \in$ [5, 10].

[1] https://github.com/mehrhardt/Multi-Modality-Imaging-with-Structure-Promoting-Regularizers

## C.3 Implementation details of GTP for CT reconstruction

**Grad-GTP.** For notation, we follow the DDBM convention that $t = 1$ at the source distribution and $t = 0$ at the target distribution. We denote the output of the learned denoiser $D_\theta$ as $\hat{\boldsymbol{x}}_0$, i.e. $\hat{\boldsymbol{x}}_0 = D_\theta(\boldsymbol{x}, t) \approx \mathbb{E}[\boldsymbol{x}_0 \mid \boldsymbol{x}_t, t]$. For Grad-GTP, we select the time-varying likelihood $g(\boldsymbol{x}, t) = (1-t) \cdot \frac{1}{2\beta^2}\|\boldsymbol{A}\hat{\boldsymbol{x}}_0 - \boldsymbol{y}\|_2^2$. This is a natural choice since $g(\boldsymbol{x}, 1) = 0$ and $g(\boldsymbol{x}, 0)$ is exactly equal to the true measurement likelihood (note that the design of the DDBM denoiser ensures that $D_\theta(\boldsymbol{x}, 0) = \boldsymbol{x}$ exactly). For numerical stability, we modify the diffusion coefficient as described in section B.2. In particular, we set the step size equal to a hyperparameter $\lambda$ divided by the norm of the gradient of $g$. This is equivalent to setting the diffusion coefficient $\sigma_t^2 = 2\lambda/(\gamma \cdot \|\nabla g(\boldsymbol{x}_t, t)\|_2)$. For sampling, we used a discretization with 1,000 timesteps. We tuned the hyperparameter $\lambda$ over the set [0.01, 0.02, 0.05, 0.1, 0.2, 0.5, 1.0, 2.0, 5.0].

**Prox-GTP (strict).** For the theoretically strict variant of Prox-GTP, we adopt the same likelihood as in Grad-GTP, namely $g(\boldsymbol{x}, t) = (1-t) \cdot \frac{1}{2\beta^2}\|\boldsymbol{A}\hat{\boldsymbol{x}}_0 - \boldsymbol{y}\|_2^2$. We treat the proximal strength as a hyperparameter $\lambda$, which amounts to setting $\sigma_t^2 = 2\lambda/\gamma$. At every iteration, the proximal optimization problem is solved using 5 steps of gradient descent. For these inner iterations, we use a learning rate of 1e-2. We also clamp the gradient values to the range $[-1, 1]$ for numerical stability. We used 1,000 sampling timesteps. We tuned the hyperparameter $\lambda$ using a grid search over the set [1e-3, 2e-3, 5e-3, 1e-2, 2e-2, 5e-2, 1e-1, 2e-1, 5e-1].

**Prox-GTP.** For the default setting of Prox-GTP, we use a proximal step that can be computed efficiently, without using gradient descent or backpropagating through the denoiser. To accomplish this, we first compute the denoised estimate $\hat{\boldsymbol{x}}_0 = D_\theta(\boldsymbol{x}, t)$. We then solve the data consistency problem using the following proximal formulation:

$$\boldsymbol{x}_{\mathsf{dc}} = \operatorname*{arg\,min}_{\boldsymbol{z}} \frac{1}{2}\|\boldsymbol{z} - \hat{\boldsymbol{x}}_0\|_2^2 + \lambda(1-t)\frac{1}{2\beta^2}\|\boldsymbol{A}\boldsymbol{z} - \boldsymbol{y}\|_2^2 = \mathsf{prox}_{\lambda g}(\hat{\boldsymbol{x}}_0), \tag{15}$$

where we define $g(\boldsymbol{x}, t) = (1-t)\frac{1}{2\beta^2}\|\boldsymbol{A}\boldsymbol{x} - \boldsymbol{y}\|_2^2$. To obtain the next iterate in the sampling process, we renoise the output of the data consistency step $\boldsymbol{x}_{\mathsf{dc}}$ to the appropriate noise level. The advantage of this formulation is that the problem in Equation (15) can be solved efficiently using the conjugate gradient (CG) method. We set the number of sampling timesteps to 200. We tuned the hyperparameters $\lambda$ and the number of CG iterations $n_{\mathsf{CG}}$ by grid search over the sets $\lambda \in$ [1e-5, 1e-4, 1e-3, 1e-2, 1e-1, 1.0, 1e1, 1e2, 1e3, 1e4] and $n_{\mathsf{CG}} \in$ [5, 10].

**Unpaired GTP variants.** We also develop *unpaired* variants of our GTP models in the case where paired source and target domain images are unavailable. Our approach is based on using deep features from a pretrained model to align the two modalities, similar to the method proposed in [22]. In particular, we train a DDBM using only data from the target domain (CT). Image features are extracted using a pretrained Segment Anything Model (SAM) encoder [23]. The output of the SAM encoder is a feature map of size $64 \times 64 \times 256$. We then perform instance normalization [24] on these features. Since the dimensionality of the SAM embeddings is still fairly large, we reduce the dimension by randomly grouping the channels into 8 sets of 32 channels and averaging the features in each set. This results in a feature map of size $64 \times 64 \times 8$. This feature map is used as the conditioning vector for the DDBM. To obtain the starting point of the bridge, we take one channel from this embedding and interpolate it to a size of $256 \times 256$.

At testing time, we do not have access to the ground truth CT image. Instead, we pass the MR

image through the SAM encoder and apply the same normalization and channel averaging procedure to obtain the starting point of the bridge and the conditioning features. Otherwise, the Grad-GTP and Prox-GTP sampling procedures are unchanged. We performed hyperparameter tuning over the same ranges as described in their respective sections above.

### C.4 Implementation details of baseline methods for PET reconstruction

**PET forward model.** We implemented the PET forward model using `parallelproj` [25]. We simulate a non-time-of-flight (TOF) PET scanner with a 2D geometry. The detector ring is a 12-sided polygon with 50 line-of-response (LOR) endpoints per side. We also incorporate a CT-based attenuation correction. A piecewise-linear attenuation correction model was employed, which computes the attenuation coefficients at 511 keV according to the following formulas:

$$\mu_{\mathsf{PET}} = 0.096\,\mathrm{cm}^{-1}\left(1 + \frac{\mathrm{HU}}{1000}\right) \quad \text{for } \mathrm{HU} < 300,$$
$$\mu_{\mathsf{PET}} = 0.081\,\mathrm{cm}^{-1}\left(1 + \frac{\mathrm{HU}}{1000}\right) \quad \text{for } \mathrm{HU} \geq 300,$$

where HU is the value of the corresponding CT image in Hounsfield units [26], [27, Ch. 4]. To incorporate this correction into the forward model, we forward project $\mu_{\mathsf{PET}}$ to obtain the "attenuation sinogram." This attenuation sinogram is multiplied pointwise with the output of the PET projector to obtain the simulated measurements.

**OSEM.** The *ordered subsets expectation maximization (OSEM)* reconstruction algorithm was implemented based on an example in the `parallelproj` documentation, which can be found at this location[2]. We tuned the number of subsets used for reconstruction $n_{\mathsf{sub}}$ and the number of iterations $n_{\mathsf{iter}}$ over the sets $n_{\mathsf{sub}} \in [2, 5, 10]$ and $n_{\mathsf{iter}} \in [1, 2, 5, 10, 20, 50]$.

**PET-DDS.** We implement PET-DDS based on the description in the paper [12] and the official code implementation[3]. PET-DDS uses a conditional diffusion model as the generative prior. We implemented this model as a DDBM. In this case, the source distribution is Gaussian, but the network is always conditioned on the corresponding CT image. In PET-DDS, the measurement likelihood is introduced by first computing the estimate of the denoised image, $\hat{\boldsymbol{x}}_0$ and then performing a data consistency update by solving the following problem:

$$\boldsymbol{x}_{\mathsf{dc}} = \arg\min_{\boldsymbol{z}} g(\boldsymbol{z}) + \lambda \|\boldsymbol{z} - \hat{\boldsymbol{x}}_0\|_2^2, \tag{16}$$

where $g$ is the negative Poisson log-likelihood, which includes the forward model. After obtaining $\boldsymbol{x}_{\mathsf{dc}}$, the next iterate is obtained by renoising $\boldsymbol{x}_{\mathsf{dc}}$, as in DDS. The problem in Equation (16) is solved using preconditioned gradient descent. We performed a grid search over the gradient descent step size $\alpha$, regularization strength $\lambda$, and number of inner iterations $n_{\mathsf{iter}}$ for solving (16) across the sets $\alpha \in$ [1e-4, 2e-4, 5e-4, 1e-3, 2e-3, 5e-3, 1e-2], $\lambda \in$ [1e0, 1e1, 1e2, 1e3], $n_{\mathsf{iter}} \in [1, 2, 5]$.

[2] https://parallelproj.readthedocs.io/en/v1.10.2/auto_examples/05_algorithms/04_run_osem_projection_data.html

[3] https://github.com/Imraj-Singh/Score-Based-Generative-Models-for-PET-Image-Reconstruction

### C.5 Implementation details of Prox-GTP for PET reconstruction

We realize Prox-GTP for PET reconstruction using our trained DDBM translation model as the prior. We also begin by computing the denoised estimate $\hat{\boldsymbol{x}}_0 = D_\theta(\boldsymbol{x}, t)$. We then perform the proximal update by solving the following problem:

$$\boldsymbol{x}_{\mathsf{dc}} = \operatorname*{arg\,min}_{\boldsymbol{z}} \; (1-t) \cdot g(\boldsymbol{z}) + \lambda \|\boldsymbol{z} - \hat{\boldsymbol{x}}_0\|_2^2, \tag{17}$$

where $g$ is the negative Poisson log-likelihood. Unlike PET-DDS, which solves a similar problem by gradient descent, we solve the problem in Equation (17) by Bregman proximal gradient descent. If we define $f(\boldsymbol{z}) = (1-t)g(\boldsymbol{z})$ and $h(\boldsymbol{z}) = \lambda\|\boldsymbol{z} - \hat{\boldsymbol{x}}_0\|_2^2$, then (17) can be solved via the iteration:

$$\boldsymbol{x}^{i+1} = \operatorname*{arg\,min}_{\boldsymbol{z}} \; \langle \nabla f(\boldsymbol{x}^i), \boldsymbol{z} \rangle + h(\boldsymbol{z}) + \frac{1}{\alpha} D_\phi(\boldsymbol{z}, \boldsymbol{x}^i), \tag{18}$$

where $D_\phi$ is the *Bregman divergence* associated with an appropriate functional $\phi$. Since the likelihood is Poisson, a natural choice for $\phi$ is the negative entropy, defined by $\phi(\boldsymbol{x}) = \sum_i x_i \log x_i - x_i$, so that $D_\phi$ is the KL divergence, i.e. $D_\phi(\boldsymbol{z}, \boldsymbol{x}^i) = \phi(\boldsymbol{z}) - \phi(\boldsymbol{x}^i) - \langle \nabla\phi(\boldsymbol{x}^i), \boldsymbol{z} - \boldsymbol{x}^i \rangle = \sum_i z_i \log(z_i/x_i^i) - z_i + x_i^i$. The subproblem in Equation (18) is solved by first performing a descent step on $f$, followed by the proximal update on $h$. We first consider the descent step on $f$, which has the form:

$$\tilde{\boldsymbol{x}} = \operatorname*{arg\,min}_{\boldsymbol{z}} \; \langle \nabla f(\boldsymbol{x}^i), \boldsymbol{z} \rangle + \frac{1}{\alpha} D_\phi(\boldsymbol{z}, \boldsymbol{x}^i). \tag{19}$$

Setting the gradient of the objective function in Equation (19) equal to 0 yields the first-order optimality condition:

$$\nabla f(\boldsymbol{x}^i) + \frac{1}{\alpha}(\nabla\phi(\tilde{\boldsymbol{x}}) - \nabla\phi(\boldsymbol{x}^i)) = 0 \tag{20}$$

$$\implies \nabla\phi(\tilde{\boldsymbol{x}}) = \nabla\phi(\boldsymbol{x}^i) - \alpha\nabla f(\boldsymbol{x}^i). \tag{21}$$

We can see that $\nabla\phi(\boldsymbol{x}) = \log \boldsymbol{x}$, where the $\log$ is applied elementwise. Therefore, from Equation (21) we obtain the following update:

$$\tilde{\boldsymbol{x}} = \boldsymbol{x}^i \cdot \exp(-\alpha\nabla f(\boldsymbol{x}^i)), \tag{22}$$

where the multiplication and exponentiation operations are performed elementwise. This update is also sometimes known as *mirror descent*. We note that in practice we also introduce preconditioning to this step. Namely, we scale $\alpha$ pointwise by $1/\boldsymbol{A}^T\mathbf{1}$, where $\mathbf{1}$ is a vector of all ones (the division is also performed elementwise). After this step, the Bregman proximal step is then computed as:

$$\boldsymbol{x}^{i+1} = \operatorname*{arg\,min}_{\boldsymbol{z}} \; h(\boldsymbol{z}) + \frac{1}{\alpha} D_\phi(\boldsymbol{z}, \tilde{\boldsymbol{x}}). \tag{23}$$

Conveniently, the problem in Equation (23) separates elementwise into a series of 1D optimization problems. Furthermore, these problems are convex, which can be seen by computing the second derivatives $h''(\boldsymbol{z}) = 2\lambda$ and $(1/\alpha \cdot D_\phi(\boldsymbol{z}, \tilde{\boldsymbol{x}}))'' = 1/(\alpha \boldsymbol{z})$. As a result, the problem in (23) can be solved efficiently using a few steps of Newton's method.

To summarize, our proposed scheme starts by computing the denoised estimate $\hat{\boldsymbol{x}}_0$. We then solve the problem in Equation (17) using several iterations of Bregman proximal gradient descent. Each iteration first consists of the "mirror descent" step in Equation (22), followed by solving the

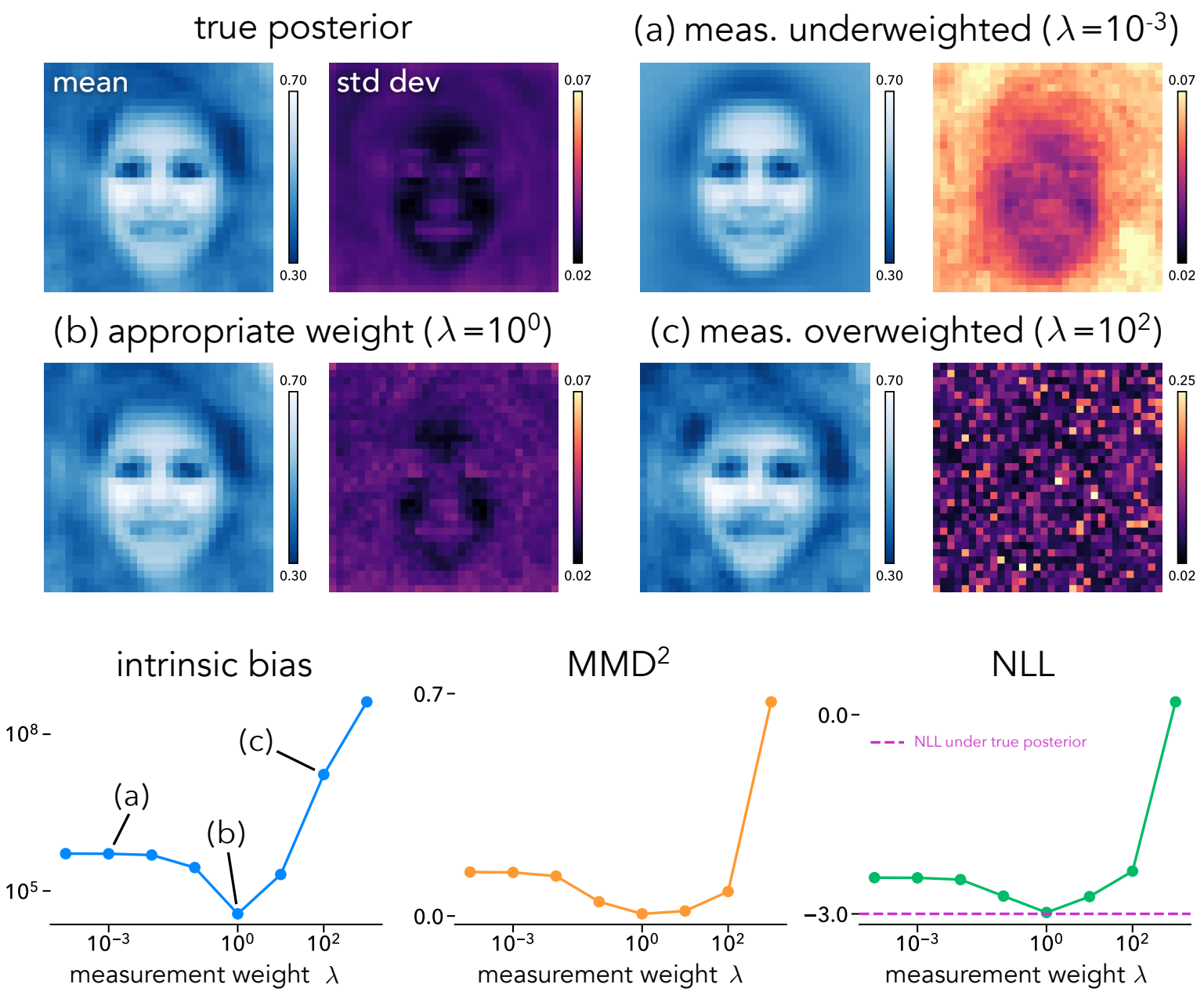


Figure 1: Numerical demonstration of using intrinsic bias as a reference-free criterion to assess posterior accuracy. Grad-GTP (Algorithm 1) is used for posterior sampling. The intrinsic bias (upper bound in Theorem 1) exhibits the same trend as other measures of posterior quality that require access to the ground truth image or posterior, including maximum mean discrepancy (MMD) and negative log-likelihood (NLL). The intrinsic bias criterion correctly detects both underweighting and overweighting of the measurement and achieves its minimum when the posterior is close to the ground truth.

optimization problem in Equation (23) using a few iterations of Newton's method. Finally, we renoise the obtained $\boldsymbol{x}_{\mathrm{dc}}$ to the appropriate noise level to obtain the next iterate.

We tuned three hyperparameters of the proposed algorithm: the regularization strength $\lambda$, the proximal step size $\alpha$, and the number of Bregman proximal gradient iterations $n_{\mathrm{B}}$. We always solve the problem in Equation (23) using 5 iterations of Newton's method. We performed a grid search over the parameters across the ranges $\alpha \in$ [1e-4, 2e-4, 5e-4, 1e-3, 2e-3, 5e-3, 1e-2], $\lambda \in$ [1e0, 1e1, 1e2, 1e3], $n_{\mathrm{B}} \in$ [1, 2, 5].

# D Additional experimental results

## D.1 Unpaired Grad-GTP and Prox-GTP

We now provide quantitative and qualitative results for our unpaired GTP algorithms, described in detail in Section C.3. Averaged values of PSNR and SSIM on the testing set are reported in Table 1 and Table 2 for SV-CT and LA-CT, respectively. In most scenarios, the unpaired GTP variant somewhat underperforms its counterpart trained with paired data. However, in challenging settings, the unpaired GTP approach outperforms diffusion-based approaches that do not use side information. For example, for 30 degree LA-CT reconstruction, the unpaired Grad-GTP outperforms DPS by 5.15 dB and the unpaired Prox-GTP outperforms DDS by 5.13 dB.

We provide qualitative comparisons to the default Grad-GTP and Prox-GTP on 8-view SV-CT and 90-degree LA-CT in Figures 2 and 3. We find that the unpaired GTP methods are able to successfully

| Method | 4 views | | 8 views | | 16 views | | 32 views | |
|---|---|---|---|---|---|---|---|---|
| | PSNR | SSIM | PSNR | SSIM | PSNR | SSIM | PSNR | SSIM |
| FBP | 13.81 | 0.333 | 17.70 | 0.292 | 22.28 | 0.290 | 27.41 | 0.373 |
| TV | 25.97 | 0.858 | 30.22 | 0.901 | 34.13 | 0.932 | 37.61 | 0.953 |
| DPS | 26.72 | 0.790 | 31.49 | 0.848 | 34.04 | 0.895 | 37.06 | 0.945 |
| DAPS | 27.27 | 0.814 | 31.52 | 0.874 | 34.22 | 0.913 | 34.95 | 0.898 |
| DDS | 28.34 | 0.835 | 33.75 | 0.898 | 37.87 | 0.949 | 39.91 | 0.960 |
| Grad-GTP | 29.76 | 0.835 | 32.43 | 0.850 | 35.12 | 0.871 | 37.35 | 0.894 |
| Prox-GTP | 30.50 | 0.885 | 34.45 | 0.916 | 37.81 | 0.952 | 39.77 | 0.956 |
| Grad-GTP (unpaired) | 27.48 | 0.774 | 30.60 | 0.793 | 33.68 | 0.834 | 37.11 | 0.911 |
| Prox-GTP (unpaired) | 28.54 | 0.841 | 32.51 | 0.885 | 36.78 | 0.941 | 39.44 | 0.954 |

Table 1: Average PSNR and SSIM values of unpaired variants of Grad-GTP and Prox-GTP on SV-CT. Results for other baseline methods are included for reference.

| Method | 30 degrees | | 60 degrees | | 90 degrees | | 120 degrees | |
|---|---|---|---|---|---|---|---|---|
| | PSNR | SSIM | PSNR | SSIM | PSNR | SSIM | PSNR | SSIM |
| FBP | 12.41 | 0.118 | 13.60 | 0.158 | 16.05 | 0.267 | 20.32 | 0.314 |
| TV | 19.88 | 0.675 | 23.60 | 0.815 | 29.49 | 0.910 | 33.58 | 0.937 |
| DPS | 20.45 | 0.685 | 22.15 | 0.731 | 27.14 | 0.835 | 31.82 | 0.900 |
| DAPS | 19.95 | 0.662 | 20.99 | 0.695 | 30.70 | 0.845 | 33.18 | 0.876 |
| DDS | 21.00 | 0.703 | 24.54 | 0.769 | 33.23 | 0.907 | 37.51 | 0.937 |
| Grad-GTP | 28.64 | 0.828 | 29.80 | 0.893 | 32.81 | 0.919 | 36.62 | 0.914 |
| Prox-GTP | 28.58 | 0.843 | 30.57 | 0.856 | 35.15 | 0.917 | 37.73 | 0.965 |
| Grad-GTP (unpaired) | 25.60 | 0.745 | 27.20 | 0.782 | 31.50 | 0.885 | 35.01 | 0.896 |
| Prox-GTP (unpaired) | 26.13 | 0.810 | 28.22 | 0.842 | 33.31 | 0.899 | 37.09 | 0.930 |

Table 2: Average PSNR and SSIM values of unpaired variants of Grad-GTP and Prox-GTP on LA-CT. Results for other baseline methods are included for reference. Note the advantage of the unpaired methods over the diffusion-based methods in the challenging 30 degree and 60 degree limited angle settings.

utilize the side information for reconstruction, even though the underlying model does not use paired training data. For example, the unpaired Prox-GTP recovers the anatomy of the kidneys in the 8-view setting, and both unpaired GTP methods correctly recover bone structures in the 90-degree LA-CT problem. This is in contrast to other diffusion-based baselines, as shown in Figures 3 and 4 in the main text.

## D.2 Prox-GTP (strict)

We present quantitative results for the variant of Prox-GTP which strictly adheres to our theoretical framework. Table 3 provides results for SV-CT, whereas Table 4 provides results for LA-CT. Across tasks and settings, we find that Prox-GTP (strict) generally performs similarly to Grad-GTP; the small gains or losses compared to Grad-GTP have no clear pattern. This result makes sense: Prox-GTP (strict) and Grad-GTP are discretizations of the same underlying dynamics, so it is natural that they produce similar results.

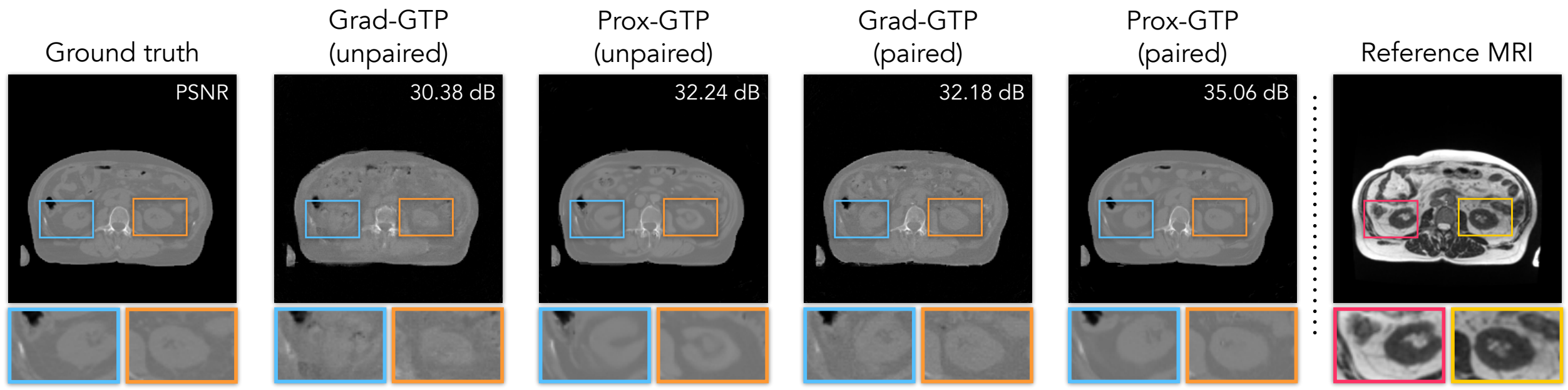


Figure 2: Qualitative comparison of reconstructions obtained using unpaired and paired GTP algorithms for 8-view SV-CT. Note how the unpaired Prox-GTP is able to recover accurate anatomical features, such as the kidneys, even though it does not use paired training data. For comparison against other baseline methods, reference Figure 3 in the main text.

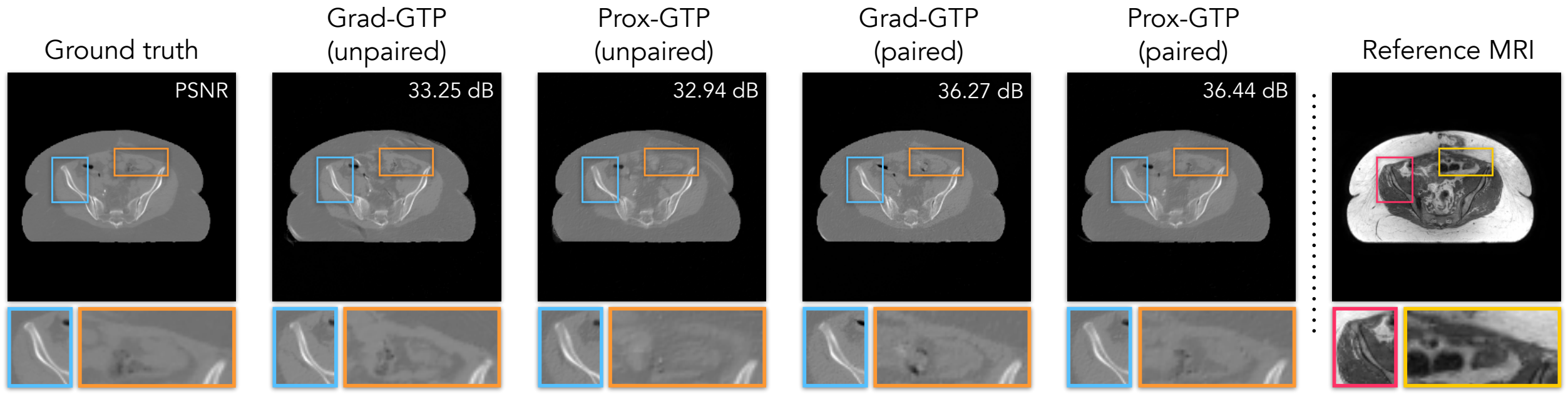


Figure 3: Visual comparison of reconstructions obtained using unpaired and paired GTP algorithms for 90-degree LA-CT. All algorithms accurately recover large structures such as bones. For comparison against other baseline methods, reference Figure 4 in the main text.

| Method | 4 views | | 8 views | | 16 views | | 32 views | |
|---|---|---|---|---|---|---|---|---|
| | PSNR | SSIM | PSNR | SSIM | PSNR | SSIM | PSNR | SSIM |
| FBP | 13.81 | 0.333 | 17.70 | 0.292 | 22.28 | 0.290 | 27.41 | 0.373 |
| TV | 25.97 | 0.858 | 30.22 | 0.901 | 34.13 | 0.932 | 37.61 | 0.953 |
| DPS | 26.72 | 0.790 | 31.49 | 0.848 | 34.04 | 0.895 | 37.06 | 0.945 |
| DAPS | 27.27 | 0.814 | 31.52 | 0.874 | 34.22 | 0.913 | 34.95 | 0.898 |
| DDS | 28.34 | 0.835 | 33.75 | 0.898 | 37.87 | 0.949 | 39.91 | 0.960 |
| Grad-GTP | 29.76 | 0.835 | 32.43 | 0.850 | 35.12 | 0.871 | 37.35 | 0.894 |
| Prox-GTP (strict) | 29.64 | 0.851 | 32.52 | 0.862 | 35.11 | 0.885 | 36.92 | 0.926 |
| Prox-GTP | 30.50 | 0.885 | 34.45 | 0.916 | 37.81 | 0.952 | 39.77 | 0.956 |

Table 3: Average PSNR and SSIM values for the strictly implemented Prox-GTP on SV-CT. Results for other baseline methods are included for reference.

## D.3 DDS-GTP and DAPS-GTP

For completeness, we tested variants of the baseline methods DAPS and DDS where the underlying diffusion model is replaced with a translation model. We call these variants DDS-GTP and DAPS-GTP. We tuned the algorithm parameters over the same ranges as DDS and DAPS; specific values are given

| Method | 30 degrees | | 60 degrees | | 90 degrees | | 120 degrees | |
|---|---|---|---|---|---|---|---|---|
| | PSNR | SSIM | PSNR | SSIM | PSNR | SSIM | PSNR | SSIM |
| FBP | 12.41 | 0.118 | 13.60 | 0.158 | 16.05 | 0.267 | 20.32 | 0.314 |
| TV | 19.88 | 0.675 | 23.60 | 0.815 | 29.49 | 0.910 | 33.58 | 0.937 |
| DPS | 20.45 | 0.685 | 22.15 | 0.731 | 27.14 | 0.835 | 31.82 | 0.900 |
| DAPS | 19.95 | 0.662 | 20.99 | 0.695 | 30.70 | 0.845 | 33.18 | 0.876 |
| DDS | 21.00 | 0.703 | 24.54 | 0.769 | 33.23 | 0.907 | 37.51 | 0.937 |
| Grad-GTP | 28.64 | 0.828 | 29.80 | 0.893 | 32.81 | 0.919 | 36.62 | 0.914 |
| Prox-GTP (strict) | 28.50 | 0.868 | 30.88 | 0.883 | 33.22 | 0.918 | 34.61 | 0.929 |
| Prox-GTP | 28.58 | 0.843 | 30.57 | 0.856 | 35.15 | 0.917 | 37.73 | 0.965 |

Table 4: Average PSNR and SSIM values for the strictly implemented Prox-GTP on LA-CT. Results for other baseline methods are included for reference.

| Method | 4 views | | 8 views | | 16 views | | 32 views | |
|---|---|---|---|---|---|---|---|---|
| | PSNR | SSIM | PSNR | SSIM | PSNR | SSIM | PSNR | SSIM |
| FBP | 13.81 | 0.333 | 17.70 | 0.292 | 22.28 | 0.290 | 27.41 | 0.373 |
| TV | 25.97 | 0.858 | 30.22 | 0.901 | 34.13 | 0.932 | 37.61 | 0.953 |
| DPS | 26.72 | 0.790 | 31.49 | 0.848 | 34.04 | 0.895 | 37.06 | 0.945 |
| DAPS | 27.27 | 0.814 | 31.52 | 0.874 | 34.22 | 0.913 | 34.95 | 0.898 |
| DDS | 28.34 | 0.835 | 33.75 | 0.898 | 37.87 | 0.949 | 39.91 | 0.960 |
| DAPS-GTP | 30.17 | 0.861 | 32.03 | 0.866 | 34.54 | 0.911 | 35.05 | 0.907 |
| DDS-GTP | 30.51 | 0.888 | 34.49 | 0.920 | 37.88 | 0.956 | 40.06 | 0.963 |
| Grad-GTP | 29.76 | 0.835 | 32.43 | 0.850 | 35.12 | 0.871 | 37.35 | 0.894 |
| Prox-GTP | 30.50 | 0.885 | 34.45 | 0.916 | 37.81 | 0.952 | 39.77 | 0.956 |

Table 5: Average PSNR and SSIM values obtained by DAPS-GTP and DDS-GTP compared to the baseline diffusion-based DAPS and DDS on SV-CT. Results for other baseline methods and the proposed Grad-GTP and Prox-GTP methods are included for reference.

in Section C.2.

Quantitative results for SV-CT reconstruction and LA-CT reconstruction can be found in Table 5 and Table 6 respectively. We find that exchanging the diffusion model for a translation model offers improved performance across almost all settings. However, we find the largest performance benefits for the most challenging problems. For example, on 30-degree LA-CT, DAPS-GTP and DDS-GTP show performance gains of more than 7 dB compared to their diffusion-based counterparts. Still, for relatively well-posed problems like 32-view SV-CT we find small gains of +0.10 dB for DAPS-GTP over DAPS and +0.15 dB for DDS-GTP over DDS. Indeed, DDS-GTP achieves the best performance of all methods in terms of PSNR in all SV-CT settings and two LA-CT settings.

### D.4 Runtime results

In Table 7 we report the average runtime per test image of various reconstruction algorithms for 8-view SV-CT and 90-degree LA-CT. In general, our GTP algorithms are well-matched in terms of runtime with their diffusion-based counterparts. The primary exception is the strict variant of Prox-GTP, which incurs the significant cost of solving the proximal step using gradient descent with

| Method | 30 degrees | | 60 degrees | | 90 degrees | | 120 degrees | |
|---|---|---|---|---|---|---|---|---|
| | PSNR | SSIM | PSNR | SSIM | PSNR | SSIM | PSNR | SSIM |
| FBP | 12.41 | 0.118 | 13.60 | 0.158 | 16.05 | 0.267 | 20.32 | 0.314 |
| TV | 19.88 | 0.675 | 23.60 | 0.815 | 29.49 | 0.910 | 33.58 | 0.937 |
| DPS | 20.45 | 0.685 | 22.15 | 0.731 | 27.14 | 0.835 | 31.82 | 0.900 |
| DAPS | 19.95 | 0.662 | 20.99 | 0.695 | 30.70 | 0.845 | 33.18 | 0.876 |
| DDS | 21.00 | 0.703 | 24.54 | 0.769 | 33.23 | 0.907 | 37.51 | 0.937 |
| DAPS-GTP | 27.38 | 0.828 | 26.33 | 0.759 | 32.98 | 0.906 | 33.66 | 0.895 |
| DDS-GTP | 28.64 | 0.853 | 30.47 | 0.889 | 35.26 | 0.933 | 37.82 | 0.966 |
| Grad-GTP | 28.64 | 0.828 | 29.80 | 0.893 | 32.81 | 0.919 | 36.62 | 0.914 |
| Prox-GTP | 28.58 | 0.843 | 30.57 | 0.856 | 35.15 | 0.917 | 37.73 | 0.965 |

Table 6: Average PSNR and SSIM values obtained by DAPS-GTP and DDS-GTP compared to the baseline diffusion-based DAPS and DDS on LA-CT. Results for other baseline methods and the proposed Grad-GTP and Prox-GTP methods are included for reference.

| Method | SV-CT (8 views) | LA-CT (90 degrees) |
|---|---|---|
| FBP | 0.01 | 0.01 |
| TV | 58.86 | 67.92 |
| JTV | 17.76 | 17.67 |
| DPS | 16.64 | 17.03 |
| DAPS | 203.94 | 193.78 |
| DDS | 12.79 | 12.93 |
| DAPS-GTP | 196.53 | 196.62 |
| DDS-GTP | 13.55 | 12.94 |
| Grad-GTP | 111.67 | 87.51 |
| Prox-GTP (strict) | 300.33 | 329.43 |
| Prox-GTP | 13.09 | 13.36 |
| Grad-GTP (unpaired) | 167.26 | 185.54 |
| Prox-GTP (unpaired) | 27.30 | 26.52 |

Table 7: Average runtime (in seconds) per test image for different reconstruction methods on SV-CT and LA-CT.

| Method | $10^5$ counts | $5 \times 10^5$ counts | $10^6$ counts |
|---|---|---|---|
| OSEM | 0.04 | 0.04 | 0.04 |
| PET-DDS | 13.15 | 13.26 | 13.31 |
| Prox-GTP | 12.14 | 13.51 | 12.69 |

Table 8: Average runtime (in seconds) per test image for different reconstruction methods across count levels for low-dose PET reconstruction.

backpropagation through the denoiser. We report runtimes for PET reconstruction across the three count levels in Table 8. We find that PET-DDS and Prox-GTP have essentially the same runtimes. This finding makes sense, given that the two methods require very similar numbers of network evaluations and evaluations of the system forward model.

### D.5 Additional visualizations

In this section, we provide visual results for additional reconstruction settings as well as for our generative priors without measurement guidance.

In Figure 4, we provide representative results on 16-view SV-CT. In this case, we see that DDS obtains a slightly higher PSNR than Prox-GTP. Still, Prox-GTP reconstructs a number of small details that are omitted by DDS. Additionally, the shapes of certain features, such as the femur on the left side of the image, are reconstructed more accurately.

In Figure 5, we show results for 60-degree LA-CT. In this setting, the proposed GTP methods demonstrate much better performance than the diffusion baselines. Every baseline method either fails to recover significant anatomical features or contains serious artifacts. Grad-GTP and Prox-GTP produce reconstructions with better anatomical consistency, especially in the missing angular region.

Figure 6 displays a number of uncurated samples from our unconditional CT diffusion model, which is implemented as a DDBM with a Gaussian source distribution. The samples were generated using 200 sampling steps. The diffusion model generates images with diverse anatomy, including chest, abdominal, and pelvic CT images. The generated images also contain many fine details and small features. These visual results confirm that our diffusion model is a high-quality generative prior.

Figure 7 shows several examples of images generated by our translation model without measurement guidance. The top row shows the reference MR images, the middle row shows the translation generated by our trained DDBM, and the bottom row shows the corresponding real CT image. Our model shows strong translation performance, generating images with features that are consistent with the MR images, and, in most cases, the ground truth CT image. These results demonstrate that DDBMs are an effective practical framework for learning translation models.

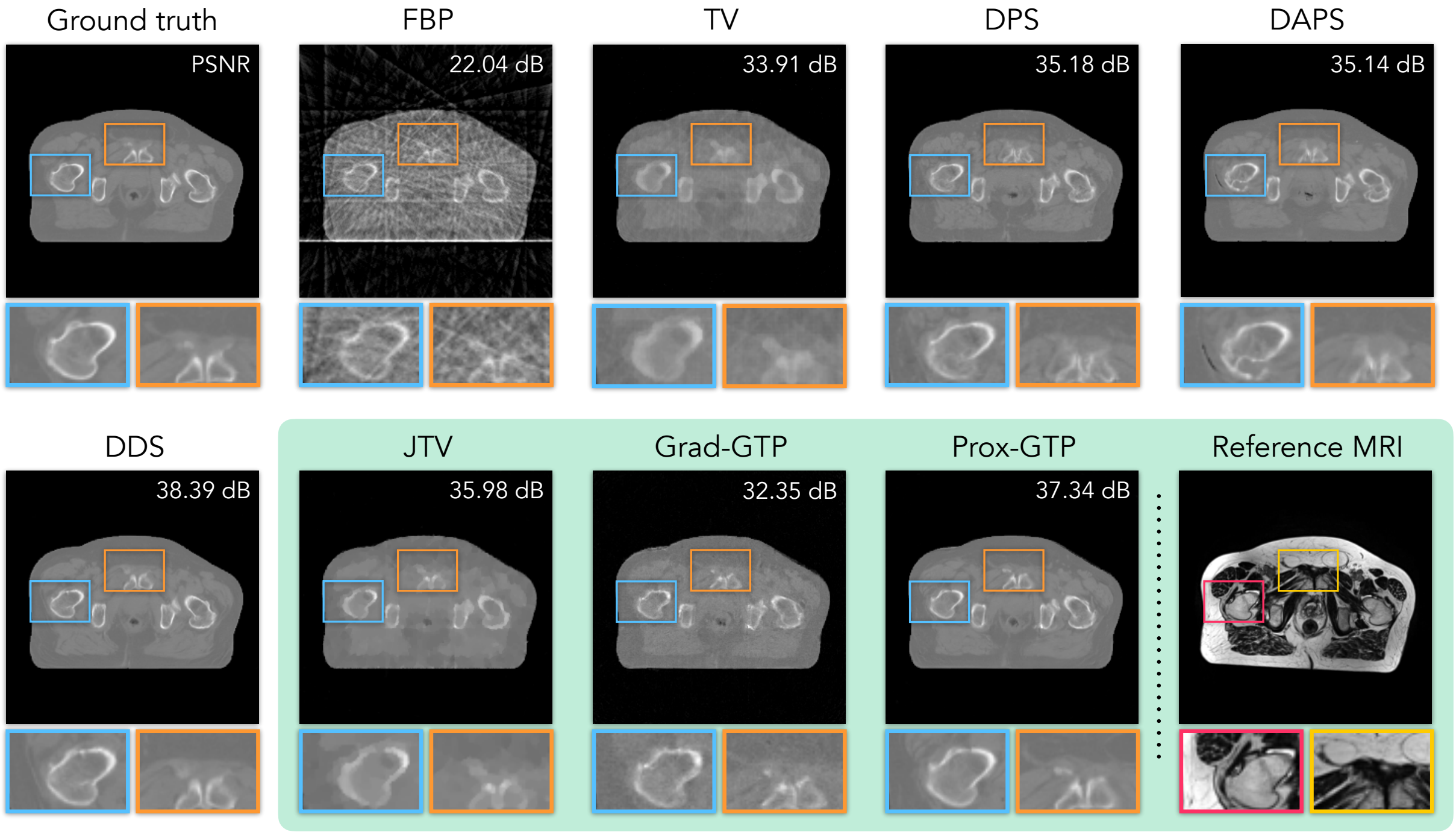


Figure 4: Visual comparison of reconstructions produced by the proposed GTP methods and baselines for 16-view SV-CT. The highlighted methods incorporate the MRI side information into the reconstruction process. Zoom-ins show regions with small anatomical details that are reconstructed more accurately by leveraging the anatomical information of the MRI.

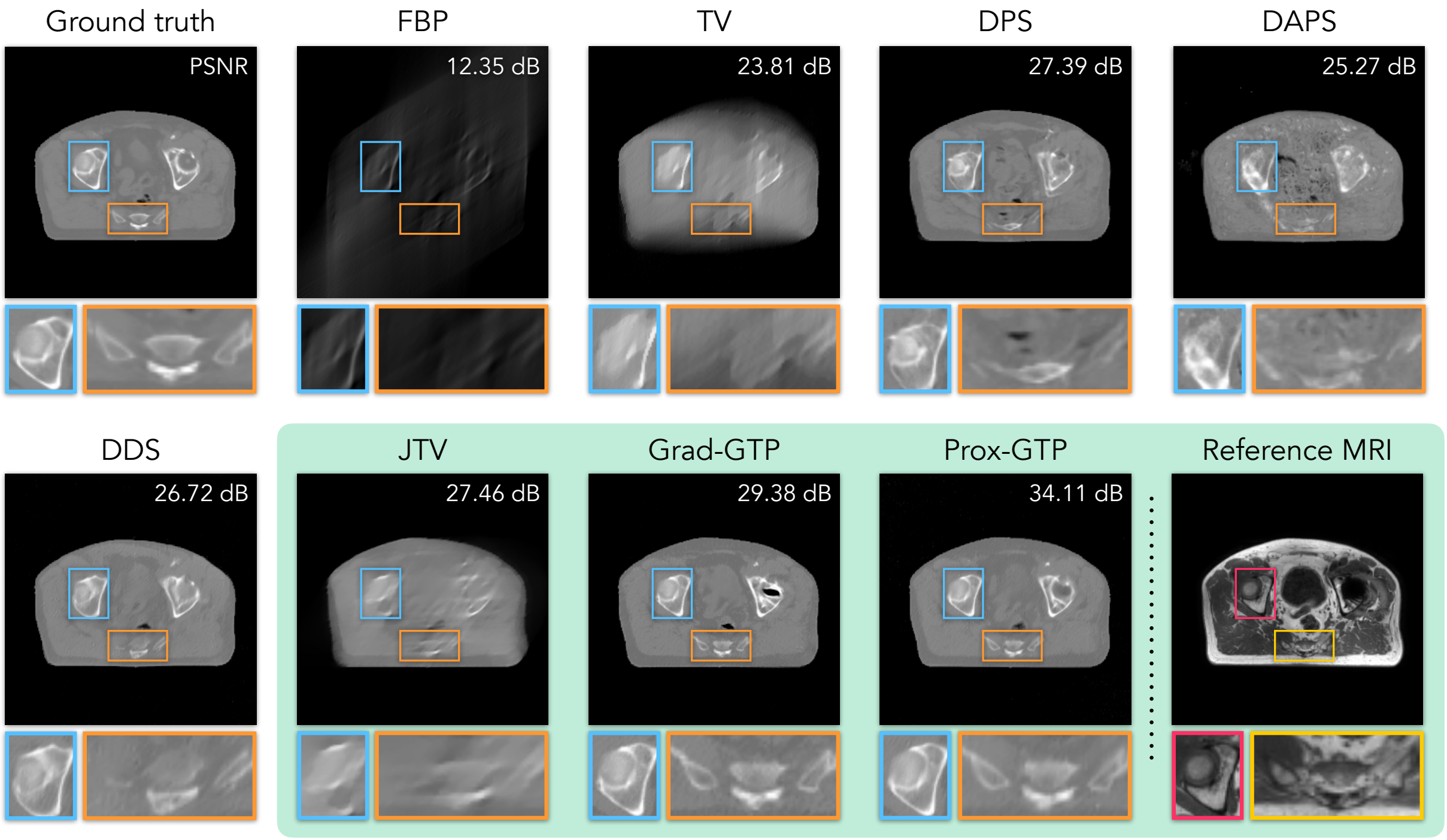


Figure 5: Qualitative comparison of reconstructions obtained using the proposed GTP methods and baselines for 60-degree LA-CT. The highlighted results are produced by methods which use the reference MR image during reconstruction. Zoom-ins highlight anatomical features that are reconstructed more accurately by GTP than the baseline methods.

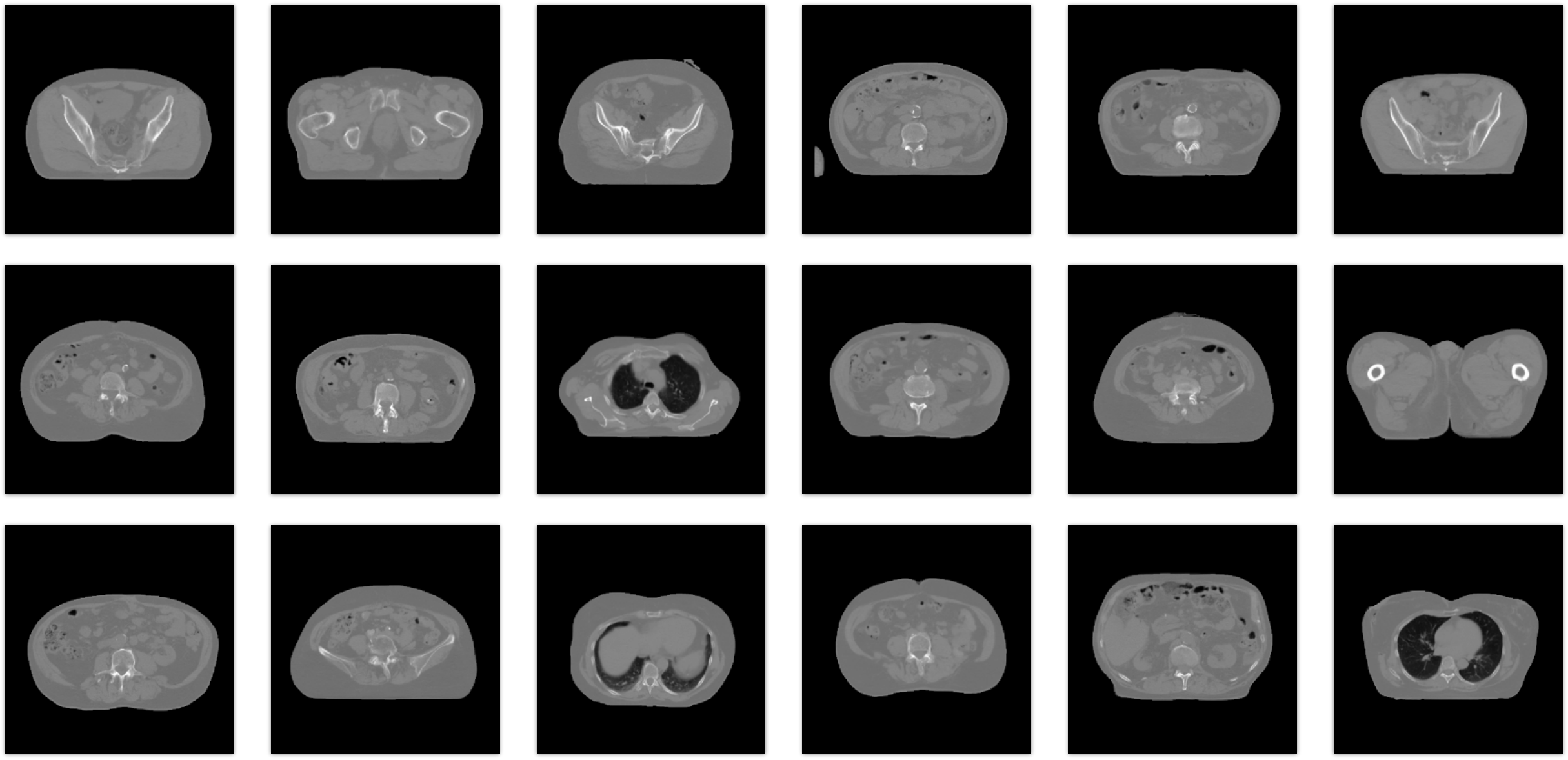

Figure 6: Visualization of uncurated samples from our unconditional CT image diffusion model, which is implemented as a DDBM. Note that the model generates images with diverse anatomical features. Zoom in to observe that the generated images contain fine details and textures.

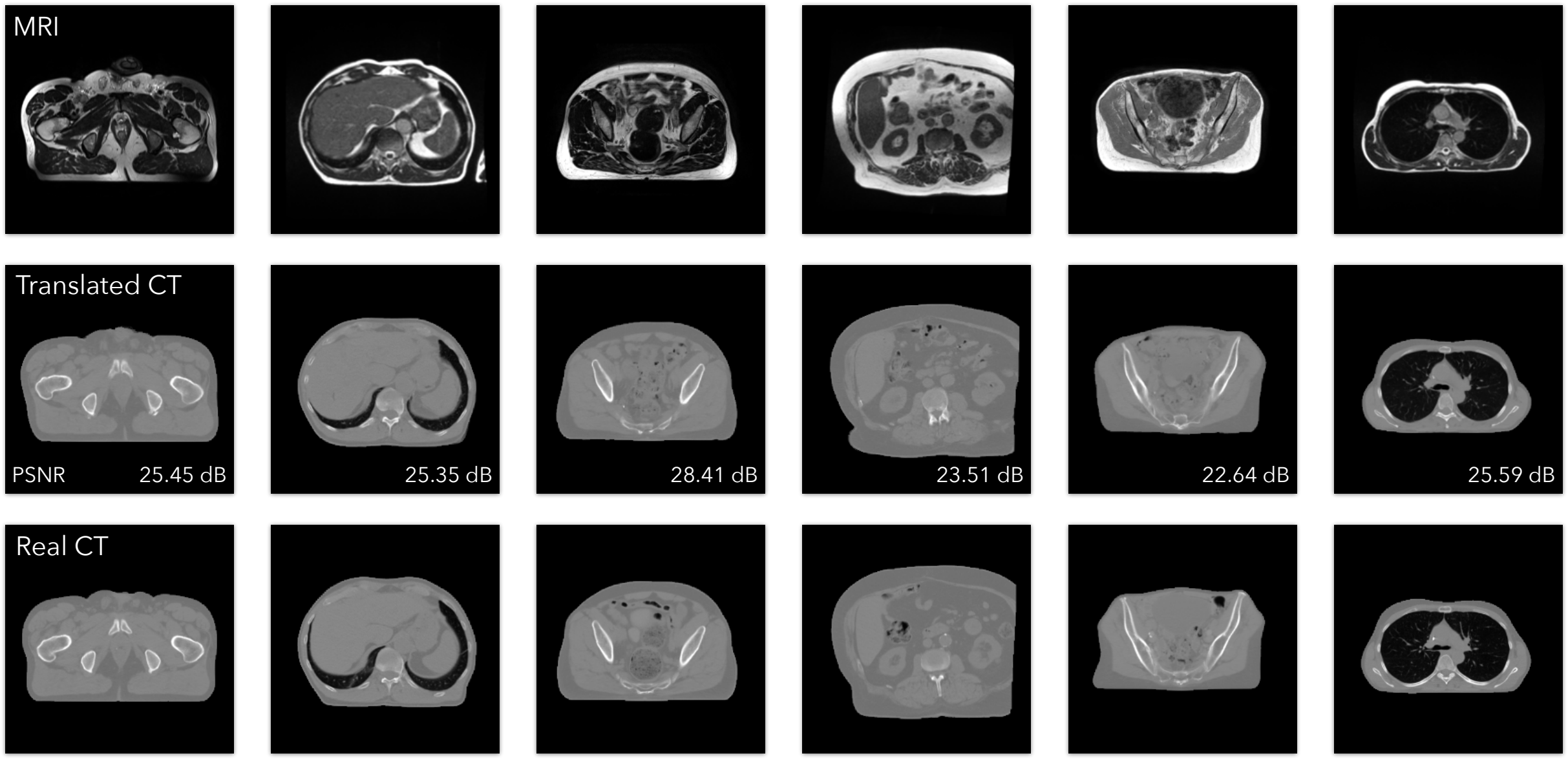


Figure 7: Visualization of translations produced by our trained DDBM model for MRI→CT translation. Top row: six reference MRI images. Middle row: synthetic CT images translated from the MR images by our generative translation model. Bottom row: corresponding real CT images. The translation model produces images that are anatomically consistent with the MR images and contain fine details.



\end{document}